\documentclass[11pt]{article}
\usepackage{tabularx}
\usepackage{longtable,supertabular,tabu}
\usepackage{amsmath}
\usepackage{amssymb}
\usepackage{amsthm}
\usepackage{latexsym}
\usepackage{enumerate}
\usepackage{mathtools}
\usepackage{setspace}
\theoremstyle{plain}
\newtheorem{theorem}{Theorem}[section]
\newtheorem{corollary}{Corollary}[section]
\newtheorem{proposition}{Proposition}[section]
\newtheorem{lemma}[theorem]{Lemma}
\newtheorem*{example}{Example}
\theoremstyle{definition}
\newtheorem{definition}{Definition}[section]

\newcommand{\tabincell}[2]{\begin{tabular}{@{}#1@{}}#2\end{tabular}}

\title{\bf Parameter-controlled inserting constructions of constant dimension subspace codes}
\author{Huimin Lao, Hao Chen, Jian Weng and Xiaoqing Tan
  \thanks{Huimin Lao, Hao Chen, Jian Weng and Xiaoqing Tan are with the College of Information Science and Technology/Cyber Security, Jinan University, Guangzhou, Guangdong Province, 510632, China, Corresponding Author: Hao Chen, haochen@jnu.edu.cn. The research of Hao Chen was supported by NSFC Grant 11531002. The research of Jian Weng was supported by NSFC Grants 61825203 and Grant U1726203. The research of Xiaoqing Tan was supported by NSFC Grant 61672014, National Cryptography Development Fund of China Grant MMJJ20180109, Natural Science Foundation of Guangdong Province of China Grant 2019A1515011069. This research was supported by the Major Program of Guangdong Basic and Applied Research under Grant 2019B03032008.}}

\begin{document}

\maketitle
\begin{abstract}
A basic problem in constant dimension subspace coding is to determine the maximal possible size ${\bf A}_q(n,d,k)$ of a set of $k$-dimensional subspaces in ${\bf F}_q^n$ such that the subspace distance satisfies $\operatorname{dis}(U,V)=2k-2\dim(U \cap V) \geq d$  for any two different $k$-dimensional subspaces $U$ and $V$ in this set. In this paper we propose new parameter-controlled inserting constructions of constant dimension subspace codes. These inserting constructions are flexible because they are controlled by parameters. Several new better lower bounds which are better than all previously constructive lower bounds can be derived from our flexible inserting constructions. $141$ new constant dimension subspace codes of distances $4,6,8$ better than previously best known codes are constructed.\\
\end{abstract}

\section{Introduction and preliminaries}

Subspace coding including constant dimension codes and mixed dimension codes has been studied extensively since the paper \cite{KK} of R. K\"{o}tter and F. R. Kschischang. The set $Grass(k,n)_q$ of all $k$-dimensional subspaces in ${\bf F}_q^n$ has $\displaystyle{n \choose k}_q=\prod_{i=0}^{k-1} \frac{q^{n-i}-1}{q^{k-i}-1}$ elements. This is the $q$-ary Gauss coefficient. The subspace distance on $Grass(k,n)_q$ can be defined by $$\operatorname{dis}(U,V)=2k-2\dim(U \cap V).$$ A set ${\bf C}$ of $M$ subspaces in $Grass(k,n)_q$, is called a $(n, M, d, k)_q$ (constant dimension $k$) subspace code if  $\operatorname{dis}(U,V) \geq d$ is satisfied for any two different subspaces $U$ and $V$ in ${\bf C}$. Sometimes we use $(n,*,d,k)_q$ to denote an CDC without counting the cardinality.\\

One main problem for the constant dimension subspace coding is to determine the maximal possible size ${\bf A}_q(n, d, k)$ of such a code for given parameters $n,d,k,q$. We refer to papers \cite{EtzionVardy,Silberstein1,Silberstein2,Silberstein3,Gluesing,Heinlein2,XuChen,CHWX,Heinlein1,LCF,Li,CKMP,Kurz} and the nice webpage \cite{table} for latest constructions and references. The presently known best constant dimension subspace codes for $n \leq 19, q\leq 9$ are listed in the table in the webpage \cite{table}. Many presently best records of constant dimension subspace codes are from the Cossidente-Kurz-Marino-Pavese combining construction in \cite{CKMP}. Though there are various good constructions \cite{XuChen,CHWX,Heinlein1,LCF,Li,CKMP,Kurz} since 2018, it seems that for many small parameter cases there are still big gaps between the presently known best upper bounds and lower bounds. This observation leads us to believe that there are some places in the present constructions which are needed to be filled with some new $k$-dimensional subspaces while the subspace distance can be preserved. In this paper it is showed that this idea works for the very effective Cossidente-Kurz-Marino-Pavese combining subspace code construction in \cite{CKMP}.  An inserting technique had been developed in Lemma 4.4 and Corollary 4.5 in \cite{CKMP} by which new $k$-dimensional subspaces can be added to the subspace codes while subspace distances can be preserved. In this paper we propose two parameter-controlled inserting constructions of constant dimension subspace codes. These two inserting constructions are the direct inserting and the multilevel type inserting, which are controlled by parameters. They are flexible by choosing different parameters to insert different positioned new $k$-dimensional subspaces in previous best known subspace codes.\\

\subsection{Rank metric codes and the Delsarte Theorem}

The rank metric on the space ${\bf M}_{a \times b}({\bf F}_q)$ of size $a \times b$ matrices over ${\bf F}_q$ is defined by the rank of matrices. The distance $d_r(A,B)$ is $\operatorname{rank}(A-B)$. The minimum rank-distance of a code ${\bf M} \subset {\bf M}_{a \times b}({\bf F}_q)$ is defined as $$d_r({\bf M})=\min_{A\neq B} \{d_r(A,B), A \in {\bf M}, B\in {\bf M} \}.$$  A rank metric code is linear if it is a linear subspace in the matrix space. For a code ${\bf M}$ in ${\bf M}_{a \times b}({\bf F}_q)$ with the minimum rank distance $d_r({\bf M}) \geq d$, it is well-known that the number of codewords in ${\bf M}$ is upper bounded by $q^{\max\{a,b\}(\min\{a,b\}-d+1)}$, we refer to \cite{Delsarte,Gabidulin,Cruz}. A rank metric code attaining this bound is called a maximum rank-distance (MRD) code. The MRD code ${\bf Q}_{q,n,t}$ consists of ${\bf F}_q$ linear mappings on ${\bf F}_q^n \cong {\bf F}_{q^n}$ defined by $q$-polynomials $a_0x+a_1x^q+\cdots+a_ix^{q^i}+\cdots+a_tx^{q^t}$, where $a_t,\ldots,a_0 \in {\bf F}_{q^n}$ are arbitrary elements in ${\bf F}_{q^n}$.  The distance of ${\bf Q}_{q,n,t}$ is $n-t$ since there are at most $q^t$ roots in ${\bf F}_{q^n}$ for each such $q$-polynomial. There are  $q^{n(t+1)}$ such $q$-polynomials in ${\bf Q}_{q,n, t}$. This kind of MRD codes has been used widely in previous constructions of constant dimension subspace codes. We refer to \cite{Gabidulin,Silberstein1,Silberstein2,RST}.\\

Let $a$ and $b$ be two positive integers. The rank distribution of a rank-metric code ${\bf M}$ in ${\bf M}_{a \times b}({\bf F}_q)$ is defined by $r(q, a, b,d, u)({\bf M})=|\{M \in {\bf M}, \operatorname{rank}(M)=u\}|$ for $u \in {\bf Z}^{+}$. The rank distribution $r(q,a,b,d,u)$ of an MRD code can be determined from its parameters. We refer the following result to Theorem 5.6 in \cite{Delsarte} or Corollary 26 in \cite{Cruz}. The Delsarte Theorem is used to count the number of codewords in a subspace code.\\

\begin{theorem}[\textbf{Delsarte 1978}]
  \label{Delsarte}
  Assume that ${\bf M} \subset {\bf M}_{a \times b}({\bf F}_q)$ is an MRD code with rank distance $d$ for $d \leq u \leq \min(a,b)$, then its rank distribution is given by $$r(q, a, b, d, u)({\bf M})=\displaystyle{\min\{a,b\} \choose u}_q \Sigma_{s=0}^{u-d} (-1)^s q^{\displaystyle{s \choose 2}} \displaystyle{u \choose s}_q (q^{\max\{a,b\}(u-s-d+1)}-1).$$
\end{theorem}

\subsection{Lifting rank metric codes}

For any given rank metric code ${\bf M}$ in ${\bf M}_{k \times n}({\bf F}_q)$ with the rank distance $d$ and cardinality $m(q, k, n, d)$, we have an $(n+k, m(q,k,n,d), 2d, k)_q$ CDC consisting of $m(q,k,n,d)$ subspaces of dimension $k$  in ${\bf F}_q^{n+k}$ spanned  by the rows of $(I_k, A)$, where $A$ is an element in ${\bf M}$. Here $I_k$ is the $k \times k$ identity matrix. It is clear that for  $A$ and $B$, the subspaces $U_A$ and $U_B$ spanned by rows of $(I_k,A)$ and $(I_k,B)$ respectively are the same if and only if $A=B$. The intersection $U_A \cap U_B$ is the set $\{ (\alpha,\alpha A)=(\beta, \beta B): \alpha (A-B)=0, \alpha \in {\bf F}_q^k\}$. Thus $\dim(U_A \cap U_B)  \leq k-d$. The distance of this CDC is $2d$. An CDC  constructed as above is called a lifted code. When ${\bf M}$ is an MRD code, this is the lifted MRD code. The Delsarte theorem can help to count the cardinality of an CDC when lifting rank-restricted rank metric codes are used to construct this CDC.\\

\subsection{Constructions based on Ferrers diagram}


Let $U$ be a $k$-dimensional subspace in ${\bf F}_q^n$ with a $k \times n$ generator matrix ${\bf U}$.
By applying Gaussian elimination on this generator matrix there exists exactly one matrix in reduced row echelon form $\xi(U)$.
The identifying vector $i(U) \in {\bf F}_2^n$ has only non-zero positions at pivots of $\xi(U)$.
The \textit{Ferrers tableaux form} $\mathcal{F(U)}$ is obtained from $\xi(U)$ by following steps.\\
\begin{itemize}
    \item Removing each rows of $\xi(U)$ from zeros to the left of the pivot of $\xi(U)$.
    \item Remove the pivot columns of $\xi(U)$.
    \item Shifting all the remaining entries to the right.
\end{itemize}
The \textit{Ferrers diagram} of $U$, denoted by $\mathcal{F}_U$,
is defined by replacing all the entries in $\mathcal{F(U)}$ by dots.

\begin{example}
    For $U$ is a  $3$-dimensional subspace in $\mathbf{F}_q^7$,
        $$
        \xi(U) =
        \begin{pmatrix}
            1 & 1 & 0 & 0 & 1 & 1 & 1\\
            0 & 0 & 1 & 0 & 1 & 0 & 1\\
            0 & 0 & 0 & 1 & 1 & 1 & 1\\
        \end{pmatrix}.
        $$
    Then $i(U)=(1,0,1,1,0,0,0)$ and

        \begin{gather*}
            \mathcal{F(U)} = \begin{pmatrix}
                1 & 1 & 1 & 1 \\
                  & 1 & 0 & 1 \\
                  & 1 & 1 & 1 \\
               \end{pmatrix}, \enspace
               \mathcal{F}_U = \begin{pmatrix}
                \bullet  & \bullet  & \bullet  & \bullet  \\
                  & \bullet  & \bullet  & \bullet  \\
                  & \bullet  & \bullet  & \bullet  \\
               \end{pmatrix}
        \end{gather*}
\end{example}

Given an Ferrers diagram $\mathcal{F}$ with $m$ dots in the rightmost
column and $l$ dots in the top row, a rank metric code in $\mathbf{F}_q^{m \times l}$ is a \textit{Ferrers diagram rank-metric code}(FDRM)$\left[\mathcal{F}, \gamma, d_f\right]$
if all the entries of its codewords not in $\mathcal{F}$ are zero, the minimal distance is $d_f$ and the dimension of the code is $\gamma$.
The number of codewords in this code is upper bounded by $q^{\min_i \{ w_i \}}$, where $w_i$ is the number of dots in $\mathcal{F}$ which is not contained in the first $i$ rows and the rightmost $d_f - 1 - i$ columns
for $0 \leq i \leq d_f - 1$. The FDRM code achieved the upper bound is called \textit{Ferrers diagram maximal rank-metric code}(FDMRD).\\

For another $k$-dimensional subspace $W$ in $\mathbf{F}_q^n$, if $i(U) = i(W)$, then we have $\mathcal{F}_U = \mathcal{F}_W$.
Given an identifying vector $v$ of length $n$ and weight $k$, the corresponding $k$-dimensional subspace $U$ in $\mathbf{F}_q^n$
can be constructed by lifting the FDRM code with Ferrers diagram $\mathcal{F}_U$.
The following result is the multilevel construction based on FDRM codes in \cite{Silberstein1}.\\


\begin{theorem}
    To construct an $(n,\psi,2d_f,k)_q$ CDC by multilevel construction, we follow these steps:
    \begin{enumerate}[(1)]
        \item Choose a binary code $\mathcal{B}$ in $\mathbf{F}_2^n$ satisfying each codeword with constant weight $k$, and the Hamming distance of $\mathcal{B}$ is $2d_f$.
        The cardinality of $\mathcal{B}$ is denoted by $s$.
        \item Construct the FDRM code $\left[\mathcal{F}_j, \gamma, d_f\right]$ for $1 \leq j \leq s$, where $\mathcal{F}_j$ corresponds to each codeword in $\mathcal{B}$.
        \item Construct $(n,\psi_j,2d_f,k)_q $ CDC by lifting the FDRM code $\left[\mathcal{F}_j, \gamma, d_f\right]$, for $1 \leq j \leq s$.
    \end{enumerate}
\end{theorem}

The union of CDC codes in (3) is the desired $(n,\psi,2d_f,k)_q$ CDC from the following Lemma \ref{lb:hamming lower bound}.
It is useful to calculate the subspace distance between two $k$-dimensional subspaces in ${\bf F}_q^n$.\\

\begin{lemma}
  \label{lb:hamming lower bound}
   Let $U$ and $U^{\prime}$ be two $k$-dimensional subspaces of $\mathbf{F}_q^{n}$. Then $\operatorname{dis}(U,U^{\prime}) \geq d_h(i(U),i(U^{\prime}))$. Here $d_h$ is the Hamming distance.
\end{lemma}

It was proved in \cite{Silberstein3} by graph matching on Ferrers diagrams that when $q^2+q+1 \geq n-\frac{k^2+k-6}{2}$ and in some other cases (see \cite{Silberstein3})
\begin{displaymath}
{\bf A}_q(n,2(k-1), k)\geq q^{2(n-k)}+\Sigma_{j=3}^{k-1} q^{2(n-\Sigma_{i=j}^k i)}+\displaystyle{n-\frac{k^2+k-6}{2} \choose 2}_q.
\end{displaymath}

For other lower bounds from multilevel construction based on Ferrers diagrams, we refer to \cite{Silberstein1,Silberstein2,Silberstein3}.\\

\subsection{Examples of new parameter-controlled inserting constructions}

The linkage construction \cite{Gluesing} is an useful construction. It was generalized in \cite{Heinlein} and used to give many good lower bounds for constant dimension subspace codes with small parameters.\enspace Then this construction has been extended in several ways in \cite{CHWX,Li,Heinlein1,CKMP}.\enspace The latest new lower bounds in \cite{CKMP} have been the best lower bounds for many small parameter cases in \cite{table}. When $n=15, d=4, k=5$, the known best lower bound before our parameter-controlled inserting constructions is ${\bf A}_2(15, 4,5) \geq 1252448586816$ from Corollary 4.5 in \cite{CKMP}. By inserting a $(15,1363968,4,5)_2$ subspace codes into the code with $1252447538240$ codewords constructed in Lemma 4.1 in \cite{CKMP}, we get a better new lower bound ${\bf A}_2(15,4,5) \geq 1252447538240+1363968=1252448902208$.\\

However in some small parameter cases the lower bound in \cite{CKMP} is worse than some previous constructions. For example in the case $n=12, d=4, k=6$ the presently known best lower bound ${\bf A}_2(12,4,6) \geq 1212491081$ is from Theorem 3.8 in \cite{CP2017}. The lower bound from Corollary 4.5 in \cite{CKMP}  is $1212451264$. Our construction is an {\em inserting } of a $(12, 2154496, 4,6)_2$ constant subspace code, which is constructed in Theorem \ref{ct:multi-blocks} below, into the code with $1212418496$ codewords constructed in Lemma 4.1 in \cite{CKMP}. Then new lower bound is ${\bf A}_2(12,4,6) \geq 1212418496+2154496=1214572992$.\\

\subsection{Notation}
\tabulinesep = 1mm
\addtocounter{table}{-1}
\begin{longtabu}{|X[1,c] |X[2,l]|}%
    \hline
    \textbf{Symbol}&\textbf{Representation}\\ \hline  \endfirsthead
    \multicolumn{2}{r}{continued table} \\ \hline
    \textbf{Symbol} & \textbf{Representation} \\ \hline \endhead
    $\xi(U)$  & The reduced row echelon form of a generator matrix of subspace U.  \\ \hline
    $ \mathbf{A}_q{(n,d,k)} $  & The best lower bound of constant $k$-dimensional subspace code in $\mathbf{F}_q^{n}$ and the minimal distance of the code is d.  \\ \hline
    $\mathbf{O}_{m \times n}$  & The zero matrix of size $\mathrm{m} \times \mathrm{n}.$ \\ \hline
    $\mathbf{I}_{k}$ & The identity matrix of size $\mathrm{k} \times \mathrm{k}.$ \\ \hline
    $m(q, a, b, d)$ & The cardinality of maximal rank distance code with parameter $(q, a, b, d).$  \\ \hline
    $r(q, a, b, d, u)$ &  The cardinality of rank metric code with parameter $(q, a, b, d)$,  the rank of elements in the code are $u$.  \\ \hline
    $m(q, a, b, d, u)$ &  The cardinality of rank metric code with parameter $(q, a, b, d)$,  the rank of elements in the code are at most $u$.
    $m(q,a,b,d,u) = 1+\sum_{i=d}^{u} r(q, a, b, d, i).$ \\ \hline
    $R(M)$ &  The $k$-dimensional subspace in $\mathbf{F}_q^n$ spanned by the rows of the matrix $M \in \mathbf{F}_q^{k \times n}$. \\ \hline
    $d_h(v_1,v_2)$ & The Hamming distance of identifying vector $v_1$ and $v_2$. \\ \hline
    $\operatorname{dis}(U, V)$ & The distance between the subspaces $\mathrm{U}$ and $\mathrm{V}$. \\ \hline
    $\operatorname{d_S}(D_1, D_2)$ & The minimal distance between two constant dimension subspace codes $D_1$ and $D_2$. \\ \hline
    $\operatorname{d_R}(\mathcal{M})$ & The minimal distance of rank metric code $\mathcal{M}$. \\ \hline
    $i(U)$ & The identifying vector of the reduced row echelon form of a generator matrix of subspace $U$. \\ \hline
    $\# D$ & The cardinality of $D$. \\ \hline
\end{longtabu}

\section{Direct inserting construction}

\subsection{CKMP combining construction of two blocks}

The following result of linkage using rank metric code with restricted rank in \cite{CKMP} is a generalization of the parallel linkage in Theorem 4 in \cite{CHWX}. We refer to \cite{Heinlein1} for the so-called generalized linkage construction.\\

\begin{theorem}[\cite{CKMP} Lemma 4.1]
    \label{linkage construction}
        Let $n, n_1,n_2, k$ be four positive integers satisfying $n_1\geq k$, $n_2 \geq k$ and $n_1+n_2=n$. Let $C_i$ be an $(n_i,*,d,k)_q$ CDC,
        and $\mathcal{M}_i$ be an $(k, n_i, \frac{d}{2})_q$ rank metric code for $ 1\leq  i \leq 2 $. Then $C=C^1 \cup  C^2$ is an $(n,*,d,k)_q$ CDC, where
        \begin{equation*}
            \begin{aligned}
                C^1 & =\{R\left(\xi \left( U_{1} \right) | \mathbf{M}_{2} \right) : U_1 \in C_1, \mathbf{M}_{2} \in \mathcal{M}_2 \}, \\
                C^2 & =\{R\left(\mathbf{M}_{1} |\xi \left( U_{2} \right)\right) : U_2 \in C_2, \mathbf{M}_{1} \in \mathcal{M}_1,  rank\left( \mathbf{M}_{1} \right) \leq k - \frac{d}{2} \}.
            \end{aligned}
        \end{equation*}
    In particular,
    \begin{align*}
        \mathbf{A}_{q}(n, d, k) & \geq \mathbf{A}_{q}\left(n_{1}, d , k\right) \cdot m\left(q, k, n_{2}, \frac{d}{2}\right)  \\
                & + (1+\sum_{u=\frac{d}{2}}^{k-\frac{d}{2}} r(q, k, n_{1}, \frac{d}{2}, u)) \cdot \mathbf{A}_{q}\left(n_{2}, d , k\right)\\
    \end{align*}
\end{theorem}

More codewords with $k$-dimensional can be added to the subspace codes in the above linkage type construction. These codes have the following property. There exists a special $n_2$-dimensional subspace $S_1$ and another special $n_1$-dimensional subspace $S_2$ in $\textbf{F}_q^n$, where $S_1$ and $S_2$ intersect trivially at the zero vector of ${\bf F}_q^n$. Moreover $S_1$ intersects with subspaces in $C^1$ trivially at the zero vector, and $S_2$ intersects with subspaces in $C^2$ trivially at the zero vector. The point is as follows. Suppose a $k$-dimensional subspace intersects $S_1$ and $S_2$ with subspaces of dimensions bigger than or equal to $\frac{d}{2}$, then the distances from this new subspace to codewords in $C$ are bigger than or equal to $d$, since their intersections have dimensions smaller than or equal to $k - \frac{d}{2}$. As Lemma 4.3 in Cossidente, Kurz, Marino and Pavese \cite{CKMP}, we have the following Lemma.\\

\begin{lemma}
    \label{disjoint lemma}
    For C constructed as in Theorem \ref{linkage construction}, there are an $n_2$ dimensional subspace $S_1$ that intersect trivially with codewords of  $C^1$ in $\mathbf{F}_q^{n}$,
    and an $n_1$ dimensional subspace $S_2$ that intersect trivially with codewords of $C^2$ in $\mathbf{F}_q^{n}$.
\end{lemma}
\begin{proof}
    Set $S_{1}=R\left(\mathbf O_{n_{2} \times n_{1}} \quad \mathbf I_{n_{2}} \right)$ and
    $S_{2}=R\left(\mathbf I_{n_{1}}  \quad \mathbf O_{n_{1} \times n_{2}} \right)$. For
    $W_i \in C^i$, $i=1,2$, we have
    \begin{equation*}
        \operatorname{dim}\left(W_{1}+S_{1}\right)
        =\operatorname{rank}\begin{pmatrix}
        \xi\left(U_{1}\right) & \mathbf M_{2} \\
        \mathbf O_{n_{2} \times n_{1}} & \mathbf I_{n_{2}}
        \end{pmatrix}
        =\operatorname{rank}\begin{pmatrix}
        \xi\left(U_{1}\right) & \mathbf O_{k \times n_{2}} \\
        \mathbf O_{n_{2} \times n_{1}} & \mathbf I_{n_{2}}
        \end{pmatrix}
        =k+n_{2},
    \end{equation*}
    where $U_1 \in C_1,$ $\mathbf{M}_2 \in \mathcal{M}_2$,
    and
    \begin{equation*}
        \operatorname{dim}\left(W_{2}+S_{2}\right)
        =\operatorname{rank}\begin{pmatrix}
        \mathbf M_{1} & \xi\left(U_{2}\right) \\
        \mathbf I_{n_{1}} & \mathbf O_{n_{1} \times n_{2}}
        \end{pmatrix}
        =\operatorname{rank}\begin{pmatrix}
        \mathbf O_{k \times n_{1}} & \xi\left(U_{2}\right) \\
        \mathbf I_{n_{1}} & \mathbf O_{n_{1} \times n_{2}}
        \end{pmatrix}
        =k+n_{1},
    \end{equation*}
    where $U_2 \in C_2,$ $\mathbf{M}_1 \in \mathcal{M}_1.$
    Therefore,
    \begin{gather*}
        \operatorname{dim}\left(W_{1} \cap S_{1}\right)=\operatorname{dim}\left(W_{1}\right)+\operatorname{dim}\left(S_{1}\right)-\operatorname{dim}\left(W_{1}+S_{1}\right)= 0, \\
        \operatorname{dim}\left(W_{2} \cap S_{2}\right)=\operatorname{dim}\left(W_{2}\right)+\operatorname{dim}\left(S_{2}\right)-\operatorname{dim}\left(W_{2}+S_{2}\right)= 0.
    \end{gather*}

\end{proof}

Based on Lemma \ref{disjoint lemma} we give a sufficient condition for a $k$-dimensional subspace in $\mathbf{F}_q^n$ can be inserted into the CKMP combining constrction.\\

\begin{lemma}
    \label{inserting sufficient condition}
    With the same notation used in Lemma \ref{disjoint lemma},
    suppose $U$ is a $k$-dimensional subspace in $\mathbf{F}_q^n$. If $\operatorname{dim}(U \cap S_1) \geq \frac{d}{2}$
    and $\operatorname{dim}(U \cap S_2) \geq \frac{d}{2}$, then $U$ can be added into the CDC code in Theorem\ref{linkage construction}.
\end{lemma}
\begin{proof}
    The Lemma \ref{disjoint lemma} gives that $\operatorname{dim}(W_{i} \cap S_{i})=0$ for $i=1,2$.
    It implies that $\operatorname{dim}(W_{i} \cap U) \leq k - \frac{d}{2}$ for $i=1,2$, then
    $\operatorname{dis}(U,W_1) \geq d$ and $\operatorname{dis}(U,W_2) \geq d$.
\end{proof}

In Lemma 4.4 of \cite{CKMP} Cossidente, Kurz,  Marino and Pavese gave an CDC that can be added to the code in Theorem \ref{linkage construction} based on the Lemma \ref{disjoint lemma}. The following result is their Lemma 4.4 for the case $l=2$ .\\

\begin{theorem}
    \label{ct:CKMP_lemma4_4}
    Let $n_1$, $n_2$, $a_1$, $a_2$, $g_1$ and $g_2$ be six positive integers satisfying $n_1+n_2=n , a_1+a_2=k , g_1+g_2= k - \frac{d}{2}$ and $g_i < a_i \leq n_i, k \leq n_i, a_i \geq \frac{d}{2}$, for $i=1,2$. Let
    $s$ be another positive integer. Suppose that $D_i^{j}$ is an $(n_i,d,a_i)_q$ CDC, for all $i=1,2$, $1 \leq j \leq s$. We assume
    $\operatorname{dis}( D_{i}^{j}, D_{i}^{j^{\prime}}) \geq 2a_i - 2g_i$ for $1\leq j < j^{\prime} \leq s$.\\

    Then $D=\bigcup_{j=1}^{s} D_j$ is an $(n,*, d,k)_q$ CDC, where $D_j = \{ U_{1} \times U_{2}: U_{i} \in \mathbf{D}_{i}^{j}, i=1,2 \}$,
    $\mathbf{D}_{i}^{j}$ is an embedding of $D_{i}^{j}$ in $\mathbf{F}_q^{n}$ such that
    the vectors contained in the codewords of $\mathbf{D}_{i}^{j}$ have non-zero entries only in the coordinates between  $n_{i-1}+1$ and ${n_i}$, ${n_0}=0$.
    $D \cup C$ is also an $(n,*, d,k)_q$ CDC. \\

    The cardinality of  $D$ satisfies
    $\# D \geq \Delta \cdot \prod_{i=1}^{2} m\left(q, a_{i}, n_{i}-a_{i}, \frac{d}{2}\right)$,
    where
        $\Delta = \min \left\{\gamma_{i}: 1 \leq i \leq 2 \right\}$ and $\gamma_{i}=\frac{m\left(q, a_{i}, n_{i}-a_{i}, a_{i}-g_{i}\right)}{m\left(q, a_{i}, n_{i}-a_{i}, \frac{d}{2}\right)}.$\\
\end{theorem}

In \cite{CKMP} to construct such different CDC code $D_i^{j}$,$D_i^{j^{\prime}}$ satisfying $$d_S\left( D_i^{j}, D_i^{j^{\prime}}\right) \geq 2(a_i - g_i),$$
Cossidente, Kurz,Marino and Pavese used lifted rank metric codes in Corollary 4.5. We re-present their method in the following Lemma.\\

\begin{lemma}[\textbf{Subcode Construction}]
        \label{subcode construction}
        Let $R_m$ be a $(q,a,b,d_m)$ linear rank-metric code, $\mathcal{M}$ be the $(q,a,b,d_s)$ sub code of $R_m$, where $d_s > d_m$.\\
        Then $s=\frac{m\left(q, a, b, d_m \right)}{m\left(q, a, b,d_s \right)}$ rank metric codes satisfying the following conditions can be constructed.
        \begin{itemize}
            \item  $\mathcal{M}_j $ is a $\left(q, a, b, d_s\right)$ rank metric code for all $ 1 \leq j \leq s$.
            \item  For $\mathbf{M} \in \mathcal{M}_j$, $\mathbf{M} ^{\prime} \in \mathcal{M}_{j^{\prime}}$, $\mathbf{M}  \neq \mathbf{M}^{\prime}$\\
            and $\operatorname{rank} \left( \mathbf{M}  - \mathbf{M} ^{\prime} \right) \geq d_m $, for all $1 \leq j < j^{\prime} \leq s$.
        \end{itemize}
\end{lemma}
\begin{proof}
    We take two different $\mathbf{M}, \mathbf{M}^{\prime} \in \mathcal{M}$. For each $\mathbf{M}_j \in R_m$,
    $\mathcal{M}_j$ is a $\left(q, a, b, d_s\right)$ rank-metric code defined by $\left\{\mathbf{M}_j + \mathbf{M}: \mathbf{M} \in \mathcal{M}\right\}$,
    since for two different elements $\mathbf{M}_1 = \mathbf{M}_j + \mathbf{M}, \mathbf{M}_2=\mathbf{M}_j + \mathbf{M}^{\prime}$,
    $\operatorname{rank}\left(\mathbf{M}_1 - \mathbf{M}_2 \right) = \operatorname{rank}\left(\mathbf{M} - \mathbf{M}^{\prime} \right) \geq
    d_s$.\\

    If $\mathbf{M}_j - \mathbf{M}_{j^{\prime}} \notin \mathcal{M}$, then $\mathcal{M}_j \cap \mathcal{M}_{j^{\prime}} = \emptyset$,
    since $\mathbf{M}_j + \mathbf{M} = \mathbf{M}_{j^{\prime}} + \mathbf{M}^{\prime}$ implies
    $\mathbf{M}_j - \mathbf{M}_{j^{\prime}} \in \mathcal{M}$.
    In other words, $\mathcal{M}_j$ is a coset of $\mathcal{M}$ in $R_m$,
    and there are $s=\frac{m\left(q, a, b, d_m \right)}{m\left(q,a,b,d_s\right)}$ distinct rank metric codes.
\end{proof}

We set $a = a_i, b= n_i - a_i, d_m = a_i - g_i, d_s = \frac{d}{2}$, where $a_i,n_i,g_i$ are the same as in Theorem \ref{ct:CKMP_lemma4_4}.
By lifting these rank metric codes with $I_{a_{i}}$, the desired $(n_i,*, d, a_i)_q$ CDCs can be obtained.
This Lemma is essential in our constructions.\\

\subsection{New inserted subspace codes}
In Theorem \ref{ct:CKMP_lemma4_4} $k$-dimensional subspaces in the code $D$ is spanned by the rows of matrices of the form
$
\begin{small}
    \begin{pmatrix}
        \mathbf{I}_{a_1} & \mathbf{M}_1 & \mathbf{O} & \mathbf{O} \\
        \mathbf{O} & \mathbf{O} & \mathbf{I}_{a_2} & \mathbf{M}_2 \\
    \end{pmatrix},
\end{small}
$ where $a_1+a_2=k$ and $\mathbf{M}_1, \mathbf{M}_2$ are from MRD codes.
However there are some gaps in the generator matrices that can be filled with more matrices from MRD codes to enlarge the  subspace codes.
The problem is how to fit these matrices into the generator matrices
so that the distances of subspaces spanned by the rows of the patched matrices are preserved.\\

In the following result we first consider an $(n,d,k)_q$ CDC code consisting of subspaces in $\mathbf{F}_q^n$ spanned by
$k$ rows of matrices which are concatenated by six small matrices $\mathbf{A}_1, \mathbf{A}_2, \mathbf{B}_1, \mathbf{B}_2, \mathbf{B}_3, \mathbf{B}_4$ in well-arranged positions,
where $\mathbf{A}_1$, $\mathbf{A}_2$ are identity matrices and $\mathbf{B}_i$ for $1 \leq i \leq 4$ are from suitable different rank metric codes.
Then we combine several CDCs by restricting rank distances of matrices in different rank metric codes.
To construct the desired rank metric codes, we need the subcode construction in Lemma \ref{subcode construction}.
We call this \textit{blocks construction} since it is constructed by multiple rank metric codes.\\


\begin{proposition}
    \label{ct:two-blocks}
    Let $n_1$, $n_2$, $a_1$, $a_2$, $b_1$ and $b_2$ be six positive integers satisfying $n_1+n_2=n$, $a_1+a_2=k$, $b_1+b_2 \geq \frac{d}{2}$ and
    $n_i \geq k$, $a_i \geq \frac{d}{2}$ and $1 \leq b_i \leq \frac{d}{2}$ for $i=1,2$.
    $\mathcal{M}_{1,2}\left(q,a_1,n_2-a_2,\frac{d}{2}\right)$, $\mathcal{M}_{2,1}\left(q,a_2,n_1-a_1,\frac{d}{2}\right)$ are rank metric codes.
    For another integer $s$, $\mathcal{M}_{1,1}^r\left(q,a_1,n_1-a_1,\frac{d}{2}\right)$, $\mathcal{M}_{2,2}^r\left(q,a_2,n_2-a_2,\frac{d}{2}\right)$ are rank metric codes for all $1 \leq r \leq s$.
    We assume $ \mathbf{M} \in \mathcal{M}_{i,i}^r$, $\mathbf{M}^\prime \in \mathcal{M}_{i,i}^{r^\prime}$ for all $1 \leq i \leq 2$, $1 \leq r<r^\prime \leq s$ satisfying
    $\mathbf{M} \neq \mathbf{M}^\prime$ and $\operatorname{rank}\left(\mathbf{M}-\mathbf{M}^{\prime}\right) \geq b_i$.\\

    Then $\mathcal{N}=\bigcup_{r=1}^{s} \mathcal{N}_r$ is an $(n,*, d,k)_q$ CDC, where $\mathcal{N}_r$ is consisting of the subspaces
    \begin{align*}
        \left\{\mathrm{R}
            \renewcommand{\arraystretch}{1.2}
            \begin{pmatrix}
            \mathbf I_{a_{1}} &\mathbf{M}_{1,1}&\mathbf{O}_1 &\mathbf{M}_{1,2} \\
            \mathbf{O}_2 &\mathbf{M}_{2,1}& \mathbf I_{a_{2}}&\mathbf{M}_{2,2}
            \end{pmatrix}\right\},\\
    \end{align*}

    \noindent where
    $\mathbf M_{1, 2} \in \mathcal{M}_{1, 2}$,
    $\mathbf M_{2, 1} \in \mathcal{M}_{2, 1}$, $\mathbf M_{i, i} \in \mathcal{M}_{i, i}^{r}$ for $i=1,2$, and
    $\mathbf{O}_1 = \mathbf{O}_{ {a}_{1} \times a_{2}},\enspace  \mathbf{O}_2 = \mathbf{O}_{{a}_{2} \times a_{1}}.$
\end{proposition}

\begin{proof}
    Since $\operatorname{rank}(\xi(\mathrm{W})) = k$ for all subspaces $\mathrm{W}\in \mathcal{N}_r$ for $1 \leq r\leq s$, the elements of $\mathcal{N}$ are $k$-dimensional subspaces in $\mathbf{F}_q^n$.\\

    For the distance analyse, let $\mathrm{W}_{1} \in \mathcal{N}_r,\enspace \mathrm{W}_{2} \in \mathcal{N}_{r^{\prime}}$ be two k-dimensional subspaces in $\mathbf{F}_q^n$ for $1 \leq r \leq r^{\prime} \leq s$,
    $\mathrm{W}_i$ is spanned by the rows of matrix $\mathbf{G}_i$ for $i=1,2$,
    \begin{gather*}
            \mathbf{G}_1 = \begin{pmatrix}
                \mathbf I_{a_{1}} & \mathbf{M}_{1,1}&\mathbf{O}_1 &\mathbf{M}_{1,2}\\
                \mathbf{O}_2 &\mathbf{M}_{2,1}& \mathbf I_{a_{2}}&\mathbf{M}_{2,2}
                \end{pmatrix}, \enspace
            \mathrm{W}_{1}=\mathrm{R}(\mathbf{G}_1), \\
            \mathbf{G}_2 = \begin{pmatrix} \mathbf I_{a_{1}} & \mathbf{M}_{1,1}^{\prime} & \mathbf{O}_1 & \mathbf{M}_{1,2}^{\prime} \\
                \mathbf{O}_2 & \mathbf{M}_{2,1}^{\prime} & \mathbf I_{a_{2}} & \mathbf{M}_{2,2}^{\prime}
                \end{pmatrix}, \enspace
            \mathrm{W}_{2}=\mathrm{R}(\mathbf{G}_2),
    \end{gather*}
    where $\mathbf{M}_{i,i} \in \mathcal{M}_{i,i}^r$, $\mathbf{M}_{i,i}^{\prime} \in \mathcal{M}_{i,i}^{r^{\prime}}$ for $i=1,2$,
    $\mathbf{M}_{1,2}, \mathbf{M}_{1,2}^{\prime} \in \mathcal{M}_{1,2}$, $\mathbf{M}_{2,1}, \mathbf{M}_{2,1}^{\prime} \in \mathcal{M}_{2,1}.$ Since the intersection of $\mathrm{W}_1$ and $\mathrm{W}_2$ in $\mathbf{F}_q^n$ is
    \begin{align*}
        \mathrm{W}_1 \cap \mathrm{W}_2
        = \left\{ ( \alpha_{1}, \alpha_{2}) \mathbf{G}_1
        = (\beta_{1}, \beta_{2}) \mathbf{G}_2 : \alpha_{i}, \beta_{i} \in \textbf{F}_{q}^{a_{i}}, i=1,2 \right\},
    \end{align*}
    we have
    \begin{small}
        \begin{equation*}
                \operatorname{dim}\left(\mathrm{W}_1 \cap \mathrm{W}_2\right)=\operatorname{dim}( \{
                \left(\alpha_{1}, \alpha_{2}\right):\left\{\begin{array}{c}
                    (\alpha_{1}, \alpha_{2})
                            \begin{pmatrix}
                                \mathbf{M}_{1,1}-\mathbf{M}_{1,1}^{\prime} \\
                                \mathbf{M}_{2,1}-\mathbf{M}_{2,1}^{\prime}
                            \end{pmatrix}
                        = 0
                    \\
                    (\alpha_{1},\alpha_{2})\begin{pmatrix}
                        \mathbf{M}_{1,2}-\mathbf{M}_{1,2}^{\prime} \\
                        \mathbf{M}_{2,2} -\mathbf{M}_{2,2}^{\prime}
                    \end{pmatrix}
                    =0
                    \end{array} , \alpha_{i} \in \mathbf{F}_q^{a_i},i = 1,2 \right. \} ).
        \end{equation*}
    \end{small}

    We analyse the following cases.
    If $\mathbf{M}_{1,2} \neq \mathbf{M}_{1,2}^{\prime}$, then
    \begin{small}
    \begin{align*}
        \operatorname{dim}\left(\mathrm{W}_1 \cap \mathrm{W}_2\right) &\leq
        \operatorname{dim}\left( \{ (\alpha_{1}, \alpha_{2}):(\alpha_{1}, \alpha_{2})\begin{pmatrix}
            \mathbf{M}_{1,2}-\mathbf{M}_{1,2}^{\prime} \\
            \mathbf{M}_{2,2} -\mathbf{M}_{2,2}^{\prime}
        \end{pmatrix} = 0, \alpha_{i} \in \mathbf{F}_q^{a_{i}}, i = 1,2 \} \right) \\
        & = \operatorname{dim}\left( \operatorname{kernel}\begin{pmatrix}
            \mathbf{M}_{1,2}-\mathbf{M}_{1,2}^{\prime} \\
            \mathbf{M}_{2,2} -\mathbf{M}_{2,2}^{\prime}
        \end{pmatrix}\right) \\
        & \leq k - \operatorname{rank}(\mathbf{M}_{1,2} - \mathbf{M}_{1,2}^{\prime}) \leq k - \frac{d}{2}.
    \end{align*}
    \end{small}
    \noindent Similarly if $\mathbf{M}_{2,1} \neq \mathbf{M}_{2,1}^{\prime}$,
    $\operatorname{dim}\left(W_1 \cap W_2\right) \leq k - \operatorname{rank}(\mathbf{M}_{2,1} - \mathbf{M}_{2,1}^{\prime}) \leq k - \frac{d}{2}.$\\

    When $\mathbf{M}_{1,2} = \mathbf{M}_{1,2}^{\prime}$ and $\mathbf{M}_{2,1} = \mathbf{M}_{2,1}^{\prime}$, if $r=r^{\prime}$,
    then $\mathbf{M}_{1,1}, \mathbf{M}_{1,1}^{\prime} \in \mathcal{M}_{1, 1}^{r}$, $\mathbf{M}_{2,2}, \mathbf{M}_{2,2}^{\prime} \in \mathcal{M}_{2, 2}^{r}$.
    If $\mathbf{M}_{1,1} \neq \mathbf{M}_{1,1}^{\prime}$, then
    $
        \operatorname{dim}\left(W_1 \cap W_2\right)
        \leq k - \operatorname{rank}\begin{pmatrix}
            \mathbf{M}_{1,1}-\mathbf{M}_{1,1}^{\prime}
        \end{pmatrix} \leq k - \frac{d}{2}.
    $
    If $\mathbf{M}_{1,1} = \mathbf{M}_{1,1}^{\prime}$, we have $\mathbf{M}_{2,2} \neq \mathbf{M}_{2,2}^{\prime}$, then
    $
        \operatorname{dim}\left(W_1 \cap W_2\right)
        \leq k - \operatorname{rank}\begin{pmatrix}
            \mathbf{M}_{2,2}-\mathbf{M}_{2,2}^{\prime}
        \end{pmatrix} \leq k - \frac{d}{2}.
    $
    If $r \neq r^{\prime}$, then $\mathbf{M}_{i,i} \in \mathcal{M}_{i, i}^{r}, \mathbf{M}_{i,i}^{\prime} \in \mathcal{M}_{i, i}^{{r}^{\prime}}$ for $i=1,2$.
    Since $\mathbf{M}_{1,1} \neq \mathbf{M}_{1,1}^{\prime}, \mathbf{M}_{2,2} \neq \mathbf{M}_{2,2}^{\prime}$ and $\mathbf{M}_{1,2} = \mathbf{M}_{1,2}^{\prime}, \mathbf{M}_{2,1} = \mathbf{M}_{2,1}^{\prime}$,
    it implies that
    \begin{equation*}
            \begin{aligned}
                \operatorname{dim}\left(W_1 \cap W_2\right)
                & \leq \operatorname{dim}\left( \{ \alpha_{1}:\alpha_{1}(\mathbf{M}_{1,1} - \mathbf{M}_{1,1}^{\prime}) = 0, \alpha_{1} \in \mathbf{F}_q^{a_{1}} \} \right) \\
                & + \operatorname{dim}\left( \{ \alpha_{2}:\alpha_{2}(\mathbf{M}_{2,2} - \mathbf{M}_{2,2}^{\prime}) = 0, \alpha_{2} \in \mathbf{F}_q^{a_{2}}\} \right)\\
                &=\operatorname{dim}\left(\operatorname{kernel} \left(\mathbf{M}_{1,1} - \mathbf{M}_{1,1}^{\prime}\right)\right) +\operatorname{dim}\left(\operatorname{kernel} \left(\mathbf{M}_{2,2}-\mathbf{M}_{2,2}^{\prime}\right)\right) \\
                & \leq a_{1}-b_{1}+a_{2}-b_{2} \leq k-\frac{d}{2}.
            \end{aligned} \\
    \end{equation*}
    It follows that $\operatorname{dis}\left(W_1, W_2\right) \geq d.$\\
\end{proof}

In the following result we improve the \textit{blocks construction} by not restricting to generator matrices of the form
$
\begin{small}
\begin{pmatrix}
    \mathbf I_{a_{1}} & \cdots & \cdots & \cdots \\
    \cdots & \cdots & \mathbf I_{a_{2}} & \cdots
\end{pmatrix}
\end{small}$
, but rather
using the matrices consisting of the generator matrices of $a_i$-dimensional CDCs in $\mathbf{F}_q^{n_i}$ for $i=1,2.$
In addition, the construction is inserted into the CKMP combining construction by restricting ranks of elements in some rank metric codes.\\


\begin{theorem}
    \label{ct:multi-blocks}
    Let $n_1$, $n_2$, $a_1$, $a_2$, $b_1$ and $b_2$ be six positive integers satisfying $n_1+n_2=n , a_1+a_2=k , b_1+b_2 \geq \frac{d}{2}$ and $a_i \leq t_i \leq n_i - \frac{d}{2}, n_i \geq k, a_i \geq \frac{d}{2}, 1 \leq b_i \leq \frac{d}{2}$ for $i=1,2$.
    $Q_i$ is an $\left( t_i, d, a_i\right)_q$ CDC for $i=1,2$ and $\mathcal{M}_{1,2}\left(q,a_1,n_2-t_2,\frac{d}{2}\right)$, $\mathcal{M}_{2,1}\left(q,a_2,n_1-t_1,\frac{d}{2}\right)$ are rank-metric codes.
    For another integer $s$, $\mathcal{M}_{1,1}^r\left(q,a_1,n_1-t_1,\frac{d}{2}\right)$, $\mathcal{M}_{2,2}^r\left(q,a_2,n_2-t_2,\frac{d}{2}\right)$ are rank-metric codes for all $1 \leq r \leq s$.
    We assume that $\mathbf{M} \neq \mathbf{M}^\prime$ and $rank\left(\mathbf{M}-\mathbf{M}^{\prime}\right) \geq b_i$ for 
$ \mathbf{M} \in \mathcal{M}_{i,i}^r$, $\mathbf{M}^\prime \in \mathcal{M}_{i,i}^{r^\prime}$ for all $1 \leq i \leq 2$, $1 \leq r<r^\prime \leq s$.\\

    Then $\mathcal{B} = \bigcup_{r=1}^{s} \mathcal{B}_r$ an $\left(n,*,d,k\right)_q$ CDC, where $\mathcal{B}_r$ is consisting of subspaces

    \begin{align*}
            \left\{ R
                \begin{pmatrix}
                \xi ({U}_{1}) & \mathbf{M}_{1,1} & \mathbf O_1 & \mathbf{M}_{1,2}\,  \\
                \mathbf O_2& \mathbf{M}_{2,1} & \xi ({U}_{2}) & \mathbf{M}_{2,2}
                \end{pmatrix}
            \right\}, \\
    \end{align*}

    \noindent where $\mathbf M_{i, i} \in \mathcal{M}_{i, i}^{r}$ for $i=1,2$, $U_{i} \in Q_{i}$ for $i=1,2$,
    $\mathbf M_{1, 2} \in \mathcal{M}_{1, 2}$, $\operatorname{rank}\left(\mathbf M_{1, 2}\right) \leq a_{1}-\frac{d}{2}$,
    $\mathbf M_{2, 1} \in \mathcal{M}_{2, 1}$, $\operatorname{rank}\left(\mathbf M_{2, 1}\right) \leq a_{2}-\frac{d}{2}$
    and $\mathbf O_1 = \mathbf O_{a_{1} \times t_{2}}, \mathbf O_2 = \mathbf O_{a_{2} \times t_{1}}$.\\

    Moreover, $\mathcal{B} \cup C$ is an $(n,*,d,k)_q$ CDC.

\end{theorem}

\begin{proof}
    Since for all subspaces $\mathrm{W}\in \mathcal{B}_r$ for $1 \leq r \leq s$, we have $\operatorname{rank}(\xi(U_i)) = a_i$
    for $U_i \in Q_i, i=1,2$, then $\operatorname{rank}(\xi(\mathrm{W})) = k$. The elements of $\mathcal{B}$ are $k$-dimensional subspaces in $\mathbf{F}_q^n$.\\

    Let $W_1 \in \mathcal{B}_r, W_2 \in \mathcal{B}_{r^{\prime}}$ be two $k$-dimensional subspaces in $\mathbf{F}_q^n$ for
    $1\leq r \leq r^{\prime} \leq s$. By the construction, there exists $U_i \in Q_i, \mathbf{M}_{i,i} \in \mathcal{M}_{i, i}^{r}$ for $i=1,2$
    and $\mathbf{M}_{i,j} \in \mathcal{M}_{i, j}$ for $1 \leq i,j \leq 2, i \neq j$
    such that
    \begin{align*}
        \mathbf{G}_1 = \begin{pmatrix}
            \xi \left(U_{1}\right) & \mathbf{M}_{1,1} & \mathbf{O}_1 & \mathbf{M}_{1,2} \\
            \mathbf{O}_2 & \mathbf{M}_{2,1} & \xi\left(\mathrm{U}_{2}\right) & \mathbf{M}_{2,2}
            \end{pmatrix}, \enspace
        \mathrm{W}_{1} = \mathrm{R}(\mathbf{G}_1),
    \end{align*}
    where $\operatorname{rank}(\mathbf{M}_{1,2})\leq a_1 - \frac{d}{2}, \enspace \operatorname{rank}(\mathbf{M}_{2,1})\leq a_2 - \frac{d}{2}$, there exists $U_i^{\prime} \in Q_i$, $\mathbf{M}_{i,i}^{\prime} \in \mathcal{M}_{i, i}^{r^{\prime}}$
    for $i=1,2$ and $\mathbf{M}_{i,j}^{\prime} \in \mathcal{M}_{i, j}$ for $1 \leq i,j \leq 2, i \neq j$ such that
    \begin{align*}
        \mathbf{G}_2 = \begin{pmatrix}
            \xi\left({U}_{1}^{\prime}\right) & \mathbf{M}_{1,1}^{\prime} & \mathbf{O}_1 & \mathbf{M}_{1,2}^{\prime} \\
            \mathbf{O}_2  & \mathbf{M}_{2,1}^{\prime} & \xi \left({U}_{2}^{\prime}\right) & \mathbf{M}_{2,2}^{\prime}
            \end{pmatrix}, \enspace
            \mathrm{W}_{2} = \mathrm{R}(\mathbf{G}_2),
    \end{align*}
    where
    $\operatorname{rank}(\mathbf{M}_{1,2}^{\prime}) \leq a_1 - \frac{d}{2}, \enspace \operatorname{rank}(\mathbf{M}_{2,1}^{\prime}) \leq a_2 - \frac{d}{2}$. The intersection of $W_1$ and $W_2$ in $\mathbf{F}_q^n$ is
    \begin{align*}
        W_1 \cap W_2
        = \left\{ ( \alpha_{1}, \alpha_{2}) \mathbf{G}_1
        = (\beta_{1}, \beta_{2}) \mathbf{G}_2 : \alpha_{i}, \beta_{i} \in \textbf{F}_{q}^{a_{i}}, i=1,2 \right\}.\\
    \end{align*}

    We analyse the following cases.\\

    \noindent (1) If $U_{1} \neq U_{1}^{\prime}$, then
    $
        \operatorname{dim}\left(W_1 \cap W_2\right) \leq \operatorname{dim}\left(U_1 \cap U_{1}^{\prime} \right) + a_2
         \leq a_1 - \frac{d}{2} + a_2 \leq k-\frac{d}{2}.\\
    $

    \noindent (2) If $U_{1} = U_{1}^{\prime}$ and $U_{2} \neq U_{2}^{\prime}$, then
    $
        \operatorname{dim}\left(W_1 \cap W_2\right)  \leq a_1 + \operatorname{dim}\left(U_2 \cap U_{2}^{\prime} \right)
         \leq a_1 +  a_2 - \frac{d}{2}  \leq k-\frac{d}{2}.\\
    $

    \noindent (3) If $U_{1}=U_{1}^{\prime}, U_{2}=U_{2}^{\prime}$, then $\alpha_{1} = \beta_{1}$ and $\alpha_{2} = \beta_{2}$
    since $\xi(U_{i}), \xi(U_{i}^{\prime})$ is the full rank matrix for $i=1,2$.
    Therefore we have
        \begin{equation*}
                \operatorname{dim}\left(W_1 \cap W_2\right)=\operatorname{dim}( \{
                \left(\alpha_{1}, \alpha_{2}\right):\left\{\begin{array}{c}
                    (\alpha_{1}, \alpha_{2})
                            \begin{pmatrix}
                                \mathbf{M}_{1,1}-\mathbf{M}_{1,1}^{\prime} \\
                                \mathbf{M}_{2,1}-\mathbf{M}_{2,1}^{\prime}
                            \end{pmatrix}
                        = 0
                    \\
                    (\alpha_{1},\alpha_{2})\begin{pmatrix}
                        \mathbf{M}_{1,2}-\mathbf{M}_{1,2}^{\prime} \\
                        \mathbf{M}_{2,2} -\mathbf{M}_{2,2}^{\prime}
                    \end{pmatrix}
                    =0
                    \end{array} \right. \} ),
        \end{equation*}
    where $\alpha_{i} \in \textbf{F}_{q}^{a_{i}}, i=1,2$.
    From a similar proof as Proposition \ref{ct:two-blocks} we get the conclusion.\\

    We need to prove that $\mathcal{B} \cup C$ is an $(n,*,d,k)_q$ CDC.
    Let $W_1$ be an element in $\mathcal{B}$ and $S_1$ and $S_2$ be the subspaces in Lemma \ref{disjoint lemma}.
    Then
    \begin{align*}
        \operatorname{dim}\left(W_1+S_{1}\right)
        &=\operatorname{rank}\begin{pmatrix}
        \xi\left(U_{1}\right) & \mathbf M_{1,1} & \mathbf O_1 & \mathbf M_{1,2} \\
        \mathbf O_2 & \mathbf M_{2,1} & \xi\left(U_{2}\right) & \mathbf M_{2,2} \\
        \mathbf O_3 & \mathbf O_4 & \mathbf I_{t_{2}} & \mathbf O_5 \\
        \mathbf O_6 & \mathbf O_7 & \mathbf O_8 & \mathbf I_{\left(n_{2}-t_{2}\right)}
        \end{pmatrix}\\
        &=\operatorname{rank}\begin{pmatrix}
        \xi\left(U_{1}\right) & \mathbf O_9 & \enspace \mathbf O_{10} & \mathbf O_{11} \\
        \mathbf O_2 & \mathbf M_{2,1} & \enspace \mathbf O_{12} & \mathbf O_{13} \\
        \mathbf O_3 & \mathbf O_4 & \enspace \mathbf I_{t_{2}} & \mathbf O_5 \\
        \mathbf O_6 & \mathbf O_7 & \enspace \mathbf O_8 & \enspace \mathbf I_{\left(n_{2}-t_{2}\right)}
        \end{pmatrix} \\
        &=a_{1}+n_{2}+\operatorname{rank}\left(\mathbf{M}_{2,1}\right).
    \end{align*}

    \noindent Here
    $\mathbf{M}_{1,1} \in \mathcal{M}_{1, 1}^{r}$, $\mathbf{M}_{2,2} \in \mathcal{M}_{2, 2}^{r}$, for $1 \leq r \leq s,$
    $\mathbf{M}_{1, 2} \in \mathcal{M}_{1, 2}$, $\operatorname{rank}(\mathbf{M}_{1,2}) \leq a_1 - \frac{d}{2}$,
    $\mathbf{M}_{2, 1} \in \mathcal{M}_{2, 1}$, $\operatorname{rank}(\mathbf{M}_{2,1}) \leq a_2 - \frac{d}{2}$,
    $\mathbf{O}_i$ for $1\leq i \leq 13$ are zero matrices of compatible sizes. Similarly we have
    $
        \operatorname{dim}\left(W_1 +S_{2}\right)=a_{2}+n_{1}+\operatorname{rank} \left(\mathbf{M}_{1,2}\right)
    $.
    Then we can calculate the dimensions of intersections
    \begin{gather*}
        \begin{array}{l}
            \operatorname{dim}\left(W_1  \cap S_{1}\right)=k+n_{2}-\left(a_{1}+n_{2}+\operatorname{rank}\left(\mathbf{M}_{2,1}\right)\right)
            =a_{2}-\operatorname{rank}\left(\mathbf{M}_{2,1}\right) \geq \frac{d}{2},
            \end{array}\\
        \begin{array}{l}
            \operatorname{dim}\left(W_1  \cap S_{2}\right)=k+n_{1}-\left(a_{2}+n_{1}+\operatorname{rank}\left(\mathbf{M}_{1,2}\right)\right)
            =a_{1}-\operatorname{rank}\left(\mathbf{M}_{1,2}\right) \geq \frac{d}{2},
        \end{array}
    \end{gather*}
    since $\operatorname{rank}\left(\mathbf{M}_{2,1}\right) \leq a_2-\frac{d}{2}$ and
    $\operatorname{rank}\left(\mathbf{M}_{1,2}\right) \leq a_1-\frac{d}{2}$.
    From Lemma \ref{inserting sufficient condition} we get the conclusion $\operatorname{d_S}\left(\mathcal{B},C\right) \geq d$.\\
\end{proof}


We consider the case $n=12, d=4, k=6$, $n_1=n_2=6,a_1=4,a_2=2,b_1=b_2=1,t_1=4,t_2=2$ as an example of Theorem 2.6.
Based on subcode construction in Lemma \ref{subcode construction}, we take matrix $\mathbf{M}_{1,1}$ from
$\mathcal{M}_{1,1}^r(q,4,2,2)$ subcode of $(q,4,2,1)$ MRD code for all $1 \leq r \leq s$,
matrix $\mathbf{M}_{2,2}$ from
$\mathcal{M}_{2,2}^r(q,2,4,2)$ subcode of $(q,2,4,1)$ MRD code for all $1 \leq r \leq s$, \\
where $s= \min(
    \frac{m(q,4,2,1)}{m(q,4,2,2)}, \frac{m(q,2,4,1)}{m(q,2,4,2)}
) = q^4$.
For matrix $\mathbf{M}_{1,2}$, we take it from $\mathcal{M}_{1,2}(q,4,4,2)$ MRD code with restricted rank $a_1 - \frac{d}{2}=2$.
Since $\mathcal{M}_{2,1}(q,2,2,2)$ with restricted rank $a_2 - \frac{d}{2}=0$ is zero matrix,
we take $\mathbf{M}_{2,1} = \mathbf{O}_{2 \times 2}$.
Then the lower bound of $q=2$ from Theorem 2.6 is
$$ \mathbf{A}_2\left(12,4,6\right) \geq \# C + \# \mathcal{B} = 1212418496 + 2154496 = 1214572992.$$  This is better than
$1212451264$ from Corollary 4.5 in \cite{CKMP} and the previously best known lower bound $1212491081$ from \cite{CMP}.
The new lower bounds from Theorem 2.6 for $\mathbf{A}_q(15,4,5),
\mathbf{A}_q(18,4,6), \mathbf{A}_q(18,6,6)$, $q=2,3,4,5,7,8,9$ are given in Corollary \ref{lb: detail bounds}. \\

From Theorem \ref{ct:multi-blocks} we totally obtain 92 better lower bounds of subspace codes than the lower bounds recorded in \cite{table}.
These lower bounds are for $\mathbf{A}_q(12,4,6)$, $\mathbf{A}_q(14,4,7)$, $\mathbf{A}_q(15,4,5)$,  $\mathbf{A}_2(16,4,4)$,
$ \mathbf{A}_q(16,4,5)$,$\mathbf{A}_q(16,4,8), \\ \mathbf{A}_q(17,4,5), \mathbf{A}_q(18,4,5),$ $\mathbf{A}_q(18,4,6)$, which are listed
in Table~\ref{tab:d=4} and for $\mathbf{A}_q(18,6,6)$ which are listed in Table~\ref{tab:d=6} for $q=2,3,4,5,7,8,9$.\\

If $n_1 \geq 2a_1$ and $n_2 \geq 2a_2$, we can insert more subspaces into the CDCs in Theorem \ref{ct:multi-blocks}.
These subspaces are spanned by the rows of matrix consisting of four matrices, which are from
two small CDCs and rank metric codes. But these generator matrices of small CDCs are placed in different positions
with Theorem \ref{ct:multi-blocks} such that the distances of subspaces are preserved.The result is given by the following Theorem 2.7.\\

\begin{theorem}
    \label{ct:parallel multiple blocks}
    With the same notation as Theorem \ref{ct:multi-blocks}, we assume that $n_i - t_i \geq a_i$, $b_i \leq c_i \leq a_i$ for $i=1,2$, and $c_1 + c_2 \leq k - \frac{d}{2}$.
    $\mathcal{M}_i(q,a_i,t_i,b_i,c_i)$ is a rank metric code with restricted rank $c_i$ and $D_i$ is an $(n_i - t_i, d, a_i)_q$ CDC code for $i=1,2$.
    The subset $\mathcal{E}$ of $k$-dimensional subspaces in $\mathbf{F}_q^n$ is constructed as follows.
    \begin{itemize}
        \item If $b_1 < \frac{d}{2}$ or  $b_2 < \frac{d}{2}$,
        we set $H_1 = \left\{ \mathbf{M}_1^1, \mathbf{M}_1^2, \cdots, \mathbf{M}_1^{s}\right\}$, where $\mathbf{M}_1^{r}$ is distinct arbitrary numbering element of $\mathcal{M}_1$ with restricted rank $c_1$ and
        $H_2 = \left\{ \mathbf{M}_2^1, \mathbf{M}_2^2, \cdots, \mathbf{M}_2^{s}\right\}$, where $\mathbf{M}_2^{r}$ is distinct arbitrary numbering element of $\mathcal{M}_2$ with restricted rank $c_2$,
        for $1 \leq r \leq s, s= \min( \# \mathcal{M}_1,  \# \mathcal{M}_2)$.
        Then
            $$
                \mathcal{E} =
                \left\{ \mathrm{R} \begin{pmatrix}
                    \mathbf{M}_1^{r} & \xi(U_1) & \mathbf{O}_1 & \mathbf{O}_2 \\
                    \mathbf{O}_3 & \mathbf{O}_4 & \mathbf{M}_2^{r} & \xi(U_2)
                    \end{pmatrix}\right\},
            $$
        where $\mathbf{M}_1^{r} \in H_1, \mathbf{M}_2^{r} \in H_2$ for $1 \leq r \leq s$, $U_i \in D_i$ for $i=1,2$,
        and $\mathbf{O}_i$ for $i=1,2,3,4$ are zero matrices of compatible size.

        \item If $b_1 = \frac{d}{2}$ and $b_2 = \frac{d}{2}$, then
                $$
                    \mathcal{E} = \left\{ \mathrm{R} \begin{pmatrix}
                    \mathbf{M}_1 & \xi(U_1) & \mathbf{O}_1 & \mathbf{O}_2 \\
                    \mathbf{O}_3 & \mathbf{O}_4 & \mathbf{M}_2 & \xi(U_2)
                    \end{pmatrix}\right\},
                $$
        where $\mathbf{M}_1 \in \mathcal{M}_1 , \mathbf{M}_2 \in \mathcal{M}_2$, $U_i \in D_i$ for $i=1,2$,
        and $\mathbf{O}_i$ for $i=1,2,3,4$ are zero matrices of compatible sizes.\\
    \end{itemize}

    Then $\mathcal{E}$ is an $(n,d,k)_q$ CDC code,
    the cardinality of $\mathcal{E}$ is
    \begin{align*}
        \# \mathcal{E} = \begin{cases}
            \# \mathcal{M}_1 \cdot \# \mathcal{M}_2 \cdot \# D_1 \cdot \# D_2, & \text{ if $b_1 = \frac{d}{2}$ and $b_2 = \frac{d}{2}$}, \\
            \Delta \cdot \# D_1 \cdot \# D_2, & \text{else},
        \end{cases}
    \end{align*}
    where $\Delta = \min(\# \mathcal{M}_1, \# \mathcal{M}_2).$\\

    Moreover, $\mathcal{B} \cup \mathcal{C} \cup \mathcal{E}$ is also an $(n,d,k)_q$ CDC code.
\end{theorem}
\begin{proof}
    Since for all subspaces $E \in \mathcal{E}$, we have $\operatorname{rank}(\xi(U_i)) = a_i$ for $U_i \in D_i$ and $i=1,2$,
    then $\operatorname{rank}(\xi(E))=k$. The elements in $\mathcal{E}$ are $k$-dimensional subspaces in $\mathbf{F}_q^n$.\\

    We analyse the following cases.\\

    (1) If $b_1 < \frac{d}{2}$ or  $b_2 < \frac{d}{2}$, let $W_1, W_2 \in \mathcal{E}$ be two $k$-dimensional subspaces in $\mathbf{F}_q^n$, by construction, we have
    \begin{align*}
        W_1 &= \mathrm{R}(\mathbf{G}_1), \enspace \mathbf{G}_1 = \begin{pmatrix}
            \mathbf{M}_1^{r} & \xi(U_1) & \mathbf{O}_1 & \mathbf{O}_2 \\
            \mathbf{O}_3 & \mathbf{O}_4 & \mathbf{M}_2^{r} & \xi(U_2)
        \end{pmatrix},\\
        W_2 & = \mathrm{R}(\mathbf{G}_2), \enspace \mathbf{G}_2 = \begin{pmatrix}
            \mathbf{M}_1^{r^{\prime}} & \xi(U_1^{\prime}) & \mathbf{O}_1 & \mathbf{O}_2 \\
            \mathbf{O}_3 & \mathbf{O}_4 & \mathbf{M}_2^{r^{\prime}} & \xi(U_2^{\prime})
        \end{pmatrix},
    \end{align*}
    where $\mathbf{M}_i^{r}, \mathbf{M}_i^{r^{\prime}} \in H_i$, $\operatorname{rank}(\mathbf{M}_i^{r}) \leq c_i$,
    $\operatorname{rank}(\mathbf{M}_i^{r^{\prime}}) \leq c_i$ for $i=1,2$, $1 \leq r \leq r^{\prime} \leq \Delta$ and
    $U_i$,$U_i^{\prime} \in D_i$ for $i=1,2$.
    The intersection of $W_1$ and $W_2$ is
    \begin{align*}
        W_1 \cap W_2
        = \left\{ ( \alpha_{1}, \alpha_{2}) \mathbf{G}_1
        = (\beta_{1}, \beta_{2}) \mathbf{G}_2 : \alpha_{i}, \beta_{i} \in \textbf{F}_{q}^{a_{i}}, i=1,2 \right\}.
    \end{align*}
    If $U_1 \neq U_1^{\prime}$, then
    $
            \operatorname{dim}(W_1 \cap W_2) \leq \operatorname{dim}(U_1 \cap U_1^{\prime}) + a_2
            \leq a_1 - \frac{d}{2} + a_2 = k - \frac{d}{2}.
    $
    Similarly, if $U_2 \neq U_2^{\prime}$, then $\operatorname{dim}(W_1 \cap W_2) \leq  a_1 + \operatorname{dim}(U_2 \cap U_2^{\prime}) \leq k - \frac{d}{2}.$
    It remains to analyse the case for $U_1 = U_1^{\prime}$, $U_2 = U_2^{\prime}$ and $r \neq r^{\prime}$.
    In this case, for such $\alpha_{i}, \beta_{i}$ for $i=1,2$, we have that $\alpha_{1} = \beta_{1}$ and $\alpha_{2} = \beta_{2}$ since $\xi(U_1)$ and $\xi(U_2)$ are full rank matrices.
    It implies that
    \begin{align*}
        \operatorname{dim}(W_1 \cap W_2) & \leq \operatorname{dim}( \{ \alpha_{1}:\alpha_{1}(\mathbf{M}_1^{r} - \mathbf{M}_1^{r^{\prime }}) = 0, \alpha_{1} \in \mathbf{F}_q^{a_{1}} \}) \\
        & + \operatorname{dim}( \{ \alpha_{2}:\alpha_{2}(\mathbf{M}_2^{r} - \mathbf{M}_2^{r^{\prime }}) = 0, \alpha_{2} \in \mathbf{F}_q^{a_{2}}\})\\
        &=\operatorname{dim}(\operatorname{kernel} (\mathbf{M}_1^{r} - \mathbf{M}_1^{r^{\prime }}) ) +\operatorname{dim}(\operatorname{kernel} (\mathbf{M}_2^{r} - \mathbf{M}_2^{r^{\prime }})) \\
        & \leq a_{1}-b_{1}+a_{2}-b_{2} \leq k-\frac{d}{2}.
    \end{align*}
    Thus for this case, $\operatorname{dis}(W_1, W_2) \geq d$.\\

    (2) If $b_1 = b_2 = \frac{d}{2}$, let $W_1, W_2 \in \mathcal{E}$ be two $k$-dimensional subspaces in $\mathbf{F}_q^n$, by construction, we have
    \begin{align*}
        W_1 &= R(\mathbf{G}_1), \enspace \mathbf{G}_1 = \begin{pmatrix}
            \mathbf{M}_1 & \xi(U_1) & \mathbf{O}_1 & \mathbf{O}_2 \\
            \mathbf{O}_3 & \mathbf{O}_4 & \mathbf{M}_2 & \xi(U_2)
        \end{pmatrix},  \\
        W_2 &= R(\mathbf{G}_2), \enspace \mathbf{G}_2 = \begin{pmatrix}
            \mathbf{M}_1^{\prime} & \xi(U_1^{\prime}) & \mathbf{O}_1 & \mathbf{O}_2 \\
            \mathbf{O}_3 & \mathbf{O}_4 & \mathbf{M}_2^{\prime} & \xi(U_2^{\prime})
        \end{pmatrix},
    \end{align*}
    where $\mathbf{M}_i, \mathbf{M}_i^{{\prime}} \in \mathcal{M}_i$, $\operatorname{rank}(\mathbf{M}_i) \leq c_i$,
    $\operatorname{rank}(\mathbf{M}_i^{\prime}) \leq c_i$, for $i=1,2$,
    $U_i$,$U_i^{\prime} \in D_i$ for $i=1,2$.
    Similar to the proof for the case $b_1 < \frac{d}{2}$ or $b_2 < \frac{d}{2}$,
    if $U_1 \neq U_1^{\prime}$ or $U_2 \neq U_2^{\prime}$, we have
    $\operatorname{dim}(W_1 \cap W_2) \leq k - \frac{d}{2}$.
    If $U_1 = U_1^{\prime}$ and $U_2 = U_2^{\prime}$, then $\mathbf{M}_1 \neq \mathbf{M}_1^{\prime}$ or $\mathbf{M}_2 \neq \mathbf{M}_2^{\prime}$.
    For this case, if $\mathbf{M}_1 \neq \mathbf{M}_1^{\prime}$, we have
        \begin{align*}
            \operatorname{dim}(W_1 \cap W_2) & \leq \operatorname{dim}\left( \{ \alpha_{1}:\alpha_{1}(\mathbf{M}_1 - \mathbf{M}_1^{\prime }) = 0, \alpha_{1} \in \mathbf{F}_q^{a_{1}} \} \right) + a_2\\
            &=\operatorname{dim}\left(\operatorname{kernel} \left(\mathbf{M}_1 - \mathbf{M}_1^{\prime }\right)\right) + a_2 \\
            & \leq a_{1}-b_{1}+a_{2} = k-\frac{d}{2}.
        \end{align*}
    If $\mathbf{M}_1 = \mathbf{M}_1^{\prime}$, then $\mathbf{M}_2 \neq \mathbf{M}_2^{\prime}$, we have $
        \operatorname{dim}(W_1 \cap W_2) \leq a_1 + (a_2 - b_2) = k - \frac{d}{2}.
    $
    Then in this case $\operatorname{dis}(W_1, W_2) \geq d$.
    We can calculate the cardinality of $\mathcal{E}$ directly from the proof of the above two cases.\\

    From Theorem \ref{ct:multi-blocks} $d_S(\mathcal{B},C) \geq d$. We analyse the distances of the codewords in $\mathcal{B}$ and $\mathcal{E}$.
    If $B \in \mathcal{B}$ and $E \in \mathcal{E}$ we have
    \begin{align*}
        \mathbf{G}_1 = \begin{pmatrix}
            \xi (B_{1}) & \mathbf{M}_{1,1} & \mathbf{O}_1 & \mathbf{M}_{1,2} \\
            \mathbf{O}_2 & \mathbf{M}_{2,1} & \xi(B_{2}) & \mathbf{M}_{2,2}
            \end{pmatrix}, \enspace
        B = \mathrm{R}(\mathbf{G}_1),
    \end{align*}
    where $\mathbf{M}_{1,2} \in \mathcal{M}_{1,2}$, $\operatorname{rank}(\mathbf{M}_{1,2})\leq a_1 - \frac{d}{2}$, $\mathbf{M}_{2,1} \in \mathcal{M}_{2,1}$, $\operatorname{rank}(\mathbf{M}_{2,1})\leq a_2 - \frac{d}{2}$,
    $B_i \in Q_i$ for $i=1,2$, $\mathbf{M}_{i,i} \in \mathbf{M}_{i,i}^r$ for $i=1,2$ and $1 \leq r \leq s$,
    \begin{align*}
        \mathbf{G}_2 = \begin{pmatrix}
            \mathbf{M}_1 & \xi(E_1) & \mathbf{O}_1 & \mathbf{O}_3 \\
            \mathbf{O}_2 & \mathbf{O}_4 & \mathbf{M}_2 & \xi(E_2)
        \end{pmatrix}, \enspace
        E = \mathrm{R}(\mathbf{G}_2),
    \end{align*}
    where $\mathbf{M}_{i} \in \mathcal{M}_{i}$, $\operatorname{rank}(\mathbf{M}_{i}) \leq c_i$ for $i=1,2$ and $E_i \in D_i$ for $i=1,2$.
    The intersection of $B$ and $E$ in $\mathbf{F}_q^n$ is
    \begin{align*}
        B \cap E
        = \left\{ ( \alpha_{1}, \alpha_{2}) \mathbf{G}_1
        = (\beta_{1}, \beta_{2}) \mathbf{G}_2 : \alpha_{i}, \beta_{i} \in \textbf{F}_{q}^{a_{i}}, i=1,2 \right\}.
    \end{align*}
    Since $\operatorname{rank}(\mathbf{M}_i) \leq c_i$ for $i=1,2$, and $\xi(B_i)$ is full rank matrix for $i=1,2$,
        \begin{align*}
            \operatorname{dim}(B \cap E)
            & \leq \operatorname{dim}(\{ \alpha_1: \exists \beta_1, \alpha_1 \xi(B_1) = \beta_1 \mathbf{M}_1, \alpha_{1}, \beta_{1} \in \mathbf{F}_q^{a_1} \}) \\
            & + \operatorname{dim}(\{ \alpha_2: \exists \beta_2, \alpha_2 \xi(B_2) = \beta_2 \mathbf{M}_2, \alpha_{2}, \beta_{2} \in \mathbf{F}_q^{a_2} \}) \\
            & \leq c_1 + c_2 \leq k - \frac{d}{2}.
        \end{align*}
    Then $\operatorname{dis}(B,E) \geq 2k - 2(k-\frac{d}{2}) \geq d$.\\

    It remains to analyse the distances of the codewords in $\mathcal{E}$ and $C$.
    Similar to the proof of Lemma \ref{disjoint lemma}, we can prove
    \begin{align*}
        \operatorname{dim}(E \cap S_1) &= k + n_2 - \operatorname{dim}(E + S_1) = a_2 \geq \frac{d}{2}, \\
        \operatorname{dim}(E \cap S_2) &= k + n_1 - \operatorname{dim}(E + S_2) = a_1 \geq \frac{d}{2}.
    \end{align*}
    From Lemma \ref{inserting sufficient condition}  $\operatorname{d_S}(\mathcal{E}, C) \geq d$.\\
\end{proof}

For example we consider the case $n=16,k=8,d=6$ with $n_1=n_2=8,a_1=a_2=4,b_1=2,b_2=1,c_1=3,c_2=2,t_1=t_2=4$.
Since $b_1< \frac{d}{2}=3, b_2 < \frac{d}{2}=3$, we take $\mathbf{M}_1^{r}$ from all the arbitrary numbering distinct elements
$\{ \mathbf{M}_1^{1}, \mathbf{M}_1^{2}, \cdots, \mathbf{M}_1^{s} \}$ of $\mathcal{M}_1(q,4,4,2,3)$ MRD code with rank restricted to $3$,
$\mathbf{M}_2^{r}$ from all the arbitrary numbering distinct elements $\{ \mathbf{M}_2^{1}, \mathbf{M}_2^{2}, \cdots, \mathbf{M}_2^{s} \}$ of $\mathcal{M}_2(q,4,4,1,2)$ MRD code with rank restricted to $2$,
for all $1 \leq r \leq s=\min(\# \mathcal{M}_1, \# \mathcal{M}_2 )$.
Then we have $\mathcal{E} = \min( \# \mathcal{M}_1, \# \mathcal{M}_2 ) \cdot \textbf{A}_q(4,6,4) \cdot \textbf{A}_q(4,6,4)=\min(m(q,4,4,2,3),m(q,4,4,1,2)).$
Then from Theorem \ref{ct:multi-blocks},
\begin{align*}
    \textbf{A}_2(16,6,8) & \geq \# \mathcal{C} + \# \mathcal{B} \\
    & = 282927683836352 + 1048576 =  282927684884928,
\end{align*}
which is the same as the previously best lower bound $282927684884928$ from Corollary 4.5 in \cite{CKMP}.
From Theorem \ref{ct:parallel multiple blocks}, we insert $\mathcal{E}$ with
$$
\# \mathcal{E} = \min\{ m(2,4,4,2,3), m(2,4,4,1,2) \} = \min\{2776, 7576\} = 2776
$$ codewords
to enlarge the code. This gives a better lower bound $\textbf{A}_2(16,6,8) \geq \# \mathcal{C} + \# \mathcal{B} + \# \mathcal{E} \geq 282927684887704$.
The new lower bounds from Theorem \ref{ct:parallel multiple blocks} for $\textbf{A}_q(16,6,8)$, $q=2,3,4,5,7,8,9$
are given in Corollary \ref{lb: detail bounds}.\\

This inserting construction for $t_1=a_1, t_2=a_2$ gives $28$ new lower bounds for
$\mathbf{A}_q(12,6,6), \mathbf{A}_q(16,6,8), \mathbf{A}_q(16,8,8), \mathbf{A}_q(19,6,6)$ for $q=2,3,4,5,7,8,9$, which is listed in Table \ref{tab:parallel multiple blocks d6}
and Table \ref{tab:parallel multiple blocks d8}
\footnote{For simplicity, we only consider the special case $t_1=a_1, t_2=a_2$ for avoiding too many parameters to calculate.}.\\


\section{Multilevel type inserting}
The multilevel construction and the linkage type construction are both productive constructions for constant dimension subspace codes.
In some papers these two constructions were combined to obtain better lower bounds.
In \cite{Li} a multilevel linkage construction was given. A parallel multilevel linkage type construction in \cite{LCF} was proposed as an inserting construction to the parallel linkage construction in \cite{CHWX}.
These constructions are the special case of the CKMP combining construction in Lemma 4.1 in \cite{CKMP}.\enspace However the subspace codes lifted by FDRM codes in multilevel construction can not be directly inserted into the CKMP combining construction.\\

In Lemma \ref{ct:FDRM} and Lemma \ref{ct:FDRM improved} we give a construction for an union of FDRM codes with special shaped Ferrers diagrams.
Based on this construction, a multilevel type parameter-controlled flexible inserting construction for
identifying vectors with \textit{special form} can be inserted into the CKMP combining construction.
The multilevel type inserting construction is given in Proposition \ref{multilevel blocks}.\\

Because a specification for optimal binary constant-weight code to yield best cardinality CDC in the multilevel construction is an unsolved problem \cite{Silberstein1},
we give two simple cases for the multilevel type construction in Proposition \ref{multilevel blocks} below.\enspace Our multilevel type inserting construction with two identifying vectors totally leads to $49$ better lower bounds for $\mathbf{A}_q(12,4,6),
\mathbf{A}_q(14,4,7), \\ \mathbf{A}_q(16,4,8), \mathbf{A}_q(18,4,6), \mathbf{A}_q(18,4,9), \mathbf{A}_q(18,6,9)$ and $\mathbf{A}_q(19,4,6)$.
For example, the present best lower bound $\mathbf{A}_2(18,6,9) \geq 92715451$-$56585415680$ is from Corollary 4.5 in \cite{CKMP}.
From Proposition \ref{multilevel blocks} below we have $\mathbf{A}_2(18,$ $6,9) \geq 9271545179590910976$, which is better than previously known bounds.
For $\mathbf{A}_q(12,4,6)$, $\mathbf{A}_q(14,4,7)$ and $\mathbf{A}_q(18,6,9)$, this construction improves all lower bounds from Theorem \ref{ct:multi-blocks}.\enspace All new lower bounds are listed in Table \ref{tab:mutilevel type linkage d4} and Table \ref{tab:mutilevel type linkage d6}.\enspace Our multilevel type inserting construction with multiple identifying vectors in the second case contributes $134$ better lower bounds compared with \cite{table}.\enspace For $\mathbf{A}_{2}(16,4,4), \mathbf{A}_q(18,6,6), \mathbf{A}_q(19,6,6)$ and $q=2,3,4,5,7,8,9$, this construction improves all lower bounds from Theorem \ref{ct:multi-blocks} and Theorem \ref{ct:parallel multiple blocks}.\enspace It also leads new lower bounds for $\mathbf{A}_q(10,$ $4,5)$, $\mathbf{A}_q(14,6,7)$ and $\mathbf{A}_q(18,8,9)$.
These $36$ lower bounds are listed in Table \ref{tab:mutilevel type linkage case2 d4}, Table \ref{tab:mutilevel type linkage case2 d6} and Table \ref{tab:mutilevel type linkage case2 d8}.\\

Notice that for a $k$-dimensional subspace codes in $\mathbf{F}_q^n$ with identifying vector $v$,
if the subspaces lifted by FDRM codes with Ferrers diagram $\mathcal{F}$ corresponding to $v$
satisfy the condition in Lemma \ref{inserting sufficient condition},
then the multilevel construction for such an identifying vector can be inserted into the CKMP combining construction.
There are $k$-dimensional subspaces with identifying vectors of \textit{special form} in $\mathbf{F}_q^n$ can be adapted to satisfy the condition in Lemma \ref{inserting sufficient condition}.\\
\begin{definition}\label{def:identifying vector}
    Let $n, k, d_f, \delta_1, \delta_2, u_1, u_2, \Delta$ be eight non-negative integers satisfying $\delta_1 + \delta_2 = n, u_1 + u_2 = k, u_1 \geq d_f, u_2 \geq d_f, \delta_1 \geq \Delta + u_1, \delta_2 \geq u_2 + d_f.$
    The \textit{special form} of identifying vectors $v$ is defined as
    $$
        (\underbrace{\overbrace{0 \cdots 0}^{\Delta} \overbrace{1 \cdots 1}^{u_1} 0 \cdots 0}_{\delta_1} \underbrace{\overbrace{1 \cdots 1}^{u_2} 0 \cdots 0}_{\delta_2}),
    $$
    that is, the continuous $u_1$ ones are in the first $\delta_1$ coordinates, and the first $u_2$ coordinates in the last $\delta_2$ coordinates are all ones.\\
\end{definition}

For a $k$-dimensional subspace $U$ in $\mathbf{F}_q^n$ with \textit{special form} $i(U)=v$, the Ferrers diagram $\mathcal{F}_U$ of $\mathcal{F}(U)$ is
$$
\mathcal{F}_U =
\begin{matrix}
    \overbrace{  \begin{array}{l}  \bullet \bullet \ldots \bullet \\  \bullet \bullet \ldots \bullet \\  \end{array}}^{\delta_1 -(u_1+\Delta)}  &
    \overbrace{ \begin{array}{l}  \bullet \bullet \ldots \bullet \bullet \\  \bullet \bullet \ldots \bullet \bullet \end{array}}^{\delta_2-u_2} \! \left. \begin{array}{l} \\ \\ \end{array} \right\}\small{u_1} \\
    & \; \begin{array}{l} \bullet \bullet \ldots \bullet \bullet \\  \bullet \bullet \ldots \bullet \bullet \\ \end{array} \! \left. \begin{array}{l} \\ \\ \end{array} \right\}\small{u_2}.
\end{matrix}
$$
By lifting the element of the rank metric code $\mathcal{M}$ with Ferrers diagram $\mathcal{F}_U$,
the $\xi(U)$ has form of
$$
    \begin{pmatrix}
        \mathbf{O}_1 & \mathbf{I}_{u_1} & \mathbf{M}_1 & \mathbf{O}_2 & \mathbf{M}_2 \\
        \mathbf{O}_3 & \mathbf{O}_4 & \mathbf{O}_5 & \mathbf{I}_{u_2} & \mathbf{M}_3
    \end{pmatrix},
$$
where
\begin{small}
$
    \begin{pmatrix}
        \mathbf{M}_1 & \mathbf{M}_2 \\
        \mathbf{O}_5 & \mathbf{M}_3
    \end{pmatrix} \in\mathcal{M},
$
$\mathbf{O}_{1} = \mathbf{O}_{u_1 \times \Delta},$
$\mathbf{O}_{i}$ for $i=2,3,4,5$
\end{small}
are zero matrices of compatible size.
If the rank of matrix $\mathbf{M}_2$ in $\xi(U)$ is restricted, the subspaces satisfy the condition in Lemma \ref{inserting sufficient condition}.
Thus the problem is how to construct such rank-restricted rank metric code with the Ferrers diagram $\mathcal{F}_U$.\\

In Lemma \ref{ct:FDRM} we give a construction for FDRM code with the Ferrers diagram in the shape of $\mathcal{F}_U$.
For simplicity we denote the Ferrers diagram in such \textit{special shape} by $\mathcal{F}$.
To construct the FDRM code, the intuition is that
we partition the $\mathcal{F}$ into small pieces and
use small rank metric codes. Then we recombine elements of these small rank metric codes
to form the required rank metric code.
Based on this construction, it is flexible to restrict the ranks of the small matrices in the generator matrices of the subspace.\\
\begin{definition}
    We will use the Ferrers diagram $\mathcal{F}$ of the following form.
    $$
    \mathcal{F} = \begin{matrix}
        \mathcal{F}_1 &
        \mathcal{F}_3 \\
        &
        \mathcal{F}_2
        \end{matrix},
    $$
    , with the same notation used in Definition \ref{def:identifying vector}, $\mathcal{F}_1$ is formed by dots of size of $u_1 \times (\delta_1 -(u_1+\Delta))$, $\mathcal{F}_2$ is formed by dots with size of $u_2 \times (\delta_2 - u_2)$,
    and $\mathcal{F}_3$ is formed by dots with size of $u_1 \times (\delta_2 - u_2)$.\\
\end{definition}

In the following Lemma we give a construction for FDRM code  $\left[\mathcal{F}, \gamma, d_f\right]$ with Ferrers diagram $\mathcal{F}$.
The FDRM code is constructed by several small rank metric codes, which corresponds to small Ferrers diagrams $\mathcal{F}_i$ for $i=1,2,3$.\\

\begin{lemma}
    \label{ct:FDRM}
    Let $\delta_1, \delta_2, u_1, u_2, b_1, b_2, n, k, d_f, \Delta$ be ten non-negative integers satisfying $\delta_1 + \delta_2 = n$, $u_1 + u_2 = k$, $u_1 \geq d_f$, $u_2 \geq d_f$, $\delta_1 \geq \Delta + u_1$, $\delta_2 \geq u_2 + d_f$,
    and $b_1 + b_2 \geq d_f$, $1 \leq b_1 \leq d_f$, $1 \leq b_2 \leq d_f$.
    $\mathcal{M}_1(q, u_1, \delta_1 - \Delta - u_1, b_1)$, $\mathcal{M}_2(q, u_2, \delta_2 - u_2, b_2)$,
    $\mathcal{M}_3(q, u_1, \delta_2 - u_2, d_f)$ are linear rank metric codes.
    We construct a subset $\mathcal{M}$ of $\mathbf{F}_q^{k \times (n-k-\Delta)}$.
    \begin{itemize}
        \item If $0 \leq \delta_1 - \Delta - u_1  < b_1$, we set $b_2 = d_f$ and

        $$
            \mathcal{M} = \left\{ \begin{pmatrix}
                \mathbf{O}_1 & \mathbf{M}_3 \\
                \mathbf{O}_2 & \mathbf{M}_2
            \end{pmatrix}:  \mathbf{M}_i \in \mathcal{M}_i,\enspace i=2,3 \right\},
        $$

        where $\mathbf{O}_1 = \mathbf{O}_{u_1 \times (\delta_1 - \Delta - u_1)}$ and
        $\mathbf{O}_2 = \mathbf{O}_{u_2 \times (\delta_1 - \Delta - u_1)}$.

        \item If $b_1 \leq \delta_1 - \Delta - u_1  < d_f$, we set
        $H_1 = \{ \mathbf{M}_1^1, \mathbf{M}_1^2, \cdots, \mathbf{M}_1^{s} \}$, where $\mathbf{M}_1^{r}$ is the arbitrary numbering distinct element of $\mathcal{M}_1$,
        and $H_2 = \{ \mathbf{M}_2^1, \mathbf{M}_2^2, \cdots, \mathbf{M}_2^{s} \}$ of $\mathcal{M}_2$,
        where $\mathbf{M}_2^{r}$ is the arbitrary numbering distinct element of $\mathcal{M}_2$, for $1 \leq r \leq s=\min\{ \# \mathcal{M}_1, \# \mathcal{M}_2 \}$.
        Then

        $$
            \mathcal{M} = \left\{ \begin{pmatrix}
                \mathbf{M}_1^{r} & \mathbf{M}_3 \\
                \mathbf{O}_1 & \mathbf{M}_2^{r}
            \end{pmatrix} : \mathbf{M}_i^{r} \in H_i \text{ for } i=1,2, 1 \leq r \leq s, \mathbf{M}_3 \in \mathcal{M}_3 \right\},
        $$

        where $\mathbf{O}_1 = \mathbf{O}_{u_2 \times (\delta_1 - \Delta - u_1)}$.

        \item If $d_f \leq \delta_1 - \Delta - u_1 $, we set $b_1 = b_2 = d_f$ and

        $$
        \mathcal{M} = \left\{
            \begin{pmatrix}
                \mathbf{M}_1 & \mathbf{M}_3 \\
                \mathbf{O}_1 & \mathbf{M}_2
            \end{pmatrix}: \mathbf{M}_i \in \mathcal{M}_i,\enspace i=1,2,3
        \right\},
        $$

        where $\mathbf{O}_1 = \mathbf{O}_{u_2 \times (\delta_1 - u_1 - \Delta)}$.\\
    \end{itemize}

    Then $\mathcal{M}$ is a FDRM code $\left[\mathcal{F}, * , d_f\right]$ in $\mathbf{F}_q^{k \times (n-k-\Delta)}$ ,
    with cardinality $$
        \# \mathcal{M} = \begin{cases}
            \Lambda_2 \cdot \Lambda_3 &  0 \leq \delta_1 - \Delta - u_1 < b_1, \\
           \min ( \# \mathcal{M}_1,  \# \mathcal{M}_2 ) \cdot \Lambda_3 &  b_1 \leq \delta_1 - \Delta - u_1 < d_f,\\
           \Lambda_1 \cdot \Lambda_2 \cdot \Lambda_3  & \delta_1 - \Delta - u_1 \geq d_f,
        \end{cases}
    $$
    where $\Lambda_1 = m(q, u_1, \delta_1 - \Delta - u_1, d_f)$,
    $\Lambda_2 = m(q, u_2, \delta_2 - u_2, d_f)$,
    $\Lambda_3 = m(q, u_1, \delta_2 - u_2, d_f)$.
\end{lemma}
\begin{proof}
    Let $W_1, W_2$ be different elements in $\mathcal{M}$.\\

    (1) If $0 \leq \delta_1 - \Delta - u_1 < b_1$,
    we have
    \begin{align*}
        W_1 = \begin{pmatrix}
            \mathbf{O}_2 & \mathbf{M}_3 \\
            \mathbf{O}_1 & \mathbf{M}_2
        \end{pmatrix}, \enspace
        W_2 = \begin{pmatrix}
            \mathbf{O}_2 & \mathbf{M}_3^{\prime} \\
            \mathbf{O}_1 & \mathbf{M}_2^{\prime}
        \end{pmatrix},
    \end{align*}
    where $\mathbf{M}_i, \mathbf{M}_i^{\prime} \in \mathcal{M}_i$ for $i=2,3$. If $\mathbf{M}_3 \neq \mathbf{M}_3^{\prime}$,
    $ \operatorname{d_r}(W_1, W_2) \geq
        \operatorname{rank}(
        \mathbf{M}_3 - \mathbf{M}_3^{\prime}) \geq d_f.$
    If $ \mathbf{M}_3 = \mathbf{M}_3^{\prime}$, we have $\mathbf{M}_2 \neq \mathbf{M}_2^{\prime}$,
    then $\operatorname{d_r}(W_1, W_2) \geq d_f.$
    Clearly the cardinality of $\mathcal{M}$ is given by $\# \mathcal{M}_2(q, u_2, \delta_2 - u_2, d_f) \times \# \mathcal{M}_3$.\\

    (2) If $b_1 \leq \delta_1 - \Delta - u_1 < d_f$, we have
    \begin{align*}
        W_1 = \begin{pmatrix}
            \mathbf{M}_1^{r} & \mathbf{M}_3 \\
            \mathbf{O}_1 & \mathbf{M}_2^{r}
        \end{pmatrix}, \enspace
        W_2 = \begin{pmatrix}
            \mathbf{M}_1^{r^\prime} & \mathbf{M}_3^{\prime} \\
            \mathbf{O}_1 & \mathbf{M}_2^{r^\prime}
        \end{pmatrix},
    \end{align*}
    where $1 \leq r \leq r^\prime \leq s$, $\mathbf{M}_i^{r}, \mathbf{M}_i^{r^\prime} \in H_i$ for $i=1,2$,
    and $\mathbf{M}_3, \mathbf{M}_3^{\prime} \in \mathcal{M}_3$.
    If $\mathbf{M}_3 \neq \mathbf{M}_3^{\prime}$, the proof is the same as case (1).
    If $\mathbf{M}_3 = \mathbf{M}_3^{\prime}$, we have $\mathbf{M}_1^{r} \neq \mathbf{M}_1^{r^{\prime}}$
    and $\mathbf{M}_2^{r} \neq \mathbf{M}_2^{r^{\prime}}$, then $
        \operatorname{d_r}(W_1, W_2) \geq \operatorname{rank}(\mathbf{M}_1^{r} - \mathbf{M}_1^{r^{\prime}}) +
        \operatorname{rank}(\mathbf{M}_2^{r} - \mathbf{M}_2^{r^{\prime}}) \geq b_1 + b_2 \geq d_f.
    $
    Clearly the cardinality of $\mathcal{M}$ is given by
    $\min( \# \mathcal{M}_1, \# \mathcal{M}_2) \cdot \# \mathcal{M}_3$.\\

    (3) If $d_f \leq \delta_1 - \Delta - u_1$, we have
    \begin{align*}
        W_1 = \begin{pmatrix}
            \mathbf{M}_1 & \mathbf{M}_3 \\
            \mathbf{O}_1 & \mathbf{M}_2
        \end{pmatrix}, \enspace
        W_2 = \begin{pmatrix}
            \mathbf{M}_1^{\prime} & \mathbf{M}_3^{\prime} \\
            \mathbf{O}_1 & \mathbf{M}_2^{\prime}
        \end{pmatrix},
    \end{align*}
    where $\mathbf{M}_i,\mathbf{M}_i^{\prime} \in \mathcal{M}_i$ for $i=1,2,3$.
    If $\mathbf{M}_3 \neq \mathbf{M}_3^{\prime}$, the proof is the same as case (1).
    If $\mathbf{M}_3 = \mathbf{M}_3^{\prime}$, we have $\mathbf{M}_2 \neq \mathbf{M}_2^{\prime}$ or $\mathbf{M}_1 \neq \mathbf{M}_1^{\prime}$,
    it implies that $\operatorname{d_r}(W_1, W_2) \geq d_f$.
    Clearly the cardinality of $\mathcal{M}$ is given by $\# \mathcal{M}_1(q, u_1, \delta_1 - \Delta - u_1, d_f) \cdot \# \mathcal{M}_2(q, u_2, \delta_2 - u_2, d_f) \cdot \# \mathcal{M}_3$.
    The conclusion is proved.\\
\end{proof}

If $\delta_1 - \Delta - u_1 \geq d_f$, the cardinality of $\mathcal{M}$ can be further improved by subcode construction.\\


\begin{lemma}
    \label{ct:FDRM improved}
    Let $\delta_1, \delta_2, u_1, u_2, c_1, c_2, n, k, d_f, \Delta$ be ten non-negative integers satisfying $\delta_1 + \delta_2 = n$, $u_1 + u_2 = k$, $u_1 \geq d_f, u_2 \geq d_f$, $\delta_1 \geq \Delta + u_1 + d_f$, $\delta_2 \geq u_2 + d_f$,
    and $c_1 + c_2 \geq d_f$, $1 \leq c_1 \leq d_f$, $1 \leq c_2 \leq d_f$.
    For integer $s$ and all $1 \leq j \leq s$, $\mathcal{M}_{1,j}(q, u_1, \delta_1 - \Delta - u_1, d_f)$, $\mathcal{M}_{2,j}(q, u_2, \delta_2 - u_2, d_f)$ are rank metric codes.
    $\mathcal{M}_3(q,u_1, \delta_2 - u_2, d_f )$ is another rank metric code.
    We assume  $\mathbf{M} \in \mathcal{M}_{i,j}, \mathbf{M}^{\prime} \in \mathcal{M}_{i, j^{\prime}}$ for all $1 \leq j < j^{\prime} \leq s$ and $i=1,2$ satisfying
    $\mathbf{M} \neq \mathbf{M}^{\prime}$ and
    $\operatorname{rank}(\mathbf{M} - \mathbf{M}^{\prime}) \geq c_i$.\\

    Then $\mathcal{M} = \bigcup_{j=1}^{s} \mathcal{M}_{j}$ is an $(q,k,n-k-\Delta,d_f)$ rank metric code, where
    \begin{align*}
        \mathcal{M}_{j} = \left\{ \begin{pmatrix}
            \mathbf{M}_1 & \mathbf{M}_3 \\
            \mathbf{O}_1 & \mathbf{M}_2
        \end{pmatrix}: \mathbf{M}_{i} \in \mathcal{M}_{i,j} \enspace \text{for}\enspace i=1,2, \mathbf{M}_3 \in \mathcal{M}_3 \right \}.
    \end{align*}
    The cardinality of $\mathcal{M}$ satisfies
    \begin{align*}
        \# \mathcal{M} = s & \cdot  m(q, u_1, \delta_1 - \Delta - u_1 , d_f) \cdot m(q, u_2, \delta_2 - u_2, d_f) \\
                       &\cdot m(q, u_1, \delta_2 - u_2, d_f),
    \end{align*}

    \noindent where
    $s = \min \left( \frac{m\left(q, u_1, \delta_1 - \Delta - u_1 , c_1\right)}{m\left(q, u_1,  \delta_1 - \Delta - u_1, d_f\right)},
    \frac{m\left(q, u_2, \delta_2 - u_2, c_2\right)}{m\left(q, u_2, \delta_2 - u_2, d_f\right)}
    \right).$

\end{lemma}
\begin{proof}
    Let $W_1 \in \mathcal{M}_{j}, W_2 \in \mathcal{M}_{j^{\prime}}$, by construction, we have
    \begin{align*}
        W_1 = \begin{pmatrix}
            \mathbf{M}_1 & \mathbf{M}_3 \\
            \mathbf{O}_1 & \mathbf{M}_2
        \end{pmatrix}, \enspace
        W_2 = \begin{pmatrix}
            \mathbf{M}_1^{\prime} & \mathbf{M}_3^{\prime} \\
            \mathbf{O}_1 & \mathbf{M}_2^{\prime}
        \end{pmatrix}
        \end{align*} for $\mathbf{M}_i \in \mathcal{M}_{i,j}$, $\mathbf{M}_i^{\prime} \in \mathcal{M}_{i,j^{\prime}}$ for $i=1,2$, $1 \leq j \leq j^{\prime} \leq s$.
    If $\mathbf{M}_3 \neq \mathbf{M}_3^{\prime}$, then $\operatorname{d_r}(W_1, W_2) \geq \operatorname{rank}(\mathbf{M}_3 - \mathbf{M}_3^{\prime}) \geq d_f$.
    If $\mathbf{M}_3 = \mathbf{M}_3^{\prime}$ and $j=j^{\prime}$, the proof is the same as case (3) in Lemma \ref{ct:FDRM}.
    If $\mathbf{M}_3 = \mathbf{M}_3^{\prime}$ and $j \neq j^{\prime}$,
    we have $\mathbf{M}_1 \neq  \mathbf{M}_1^{\prime}$ and $\mathbf{M}_2 \neq  \mathbf{M}_2^{\prime}$,
    then $\operatorname{d_r}(W_1, W_2) \geq \operatorname{rank}(\mathbf{M}_1 - \mathbf{M}_1^{\prime}) + \operatorname{rank}(\mathbf{M}_2 - \mathbf{M}_2^{\prime}) \geq
    c_1 + c_2 \geq d_f$.
    The cardinality of $\mathcal{M}$ can be calculated directly from the proof.
\end{proof}


Based on the construction in Lemma \ref{ct:FDRM} and Lemma \ref{ct:FDRM improved},
we obtain the new construction by inserting the multilevel type construction into the combining construction in \cite{CKMP}.\\

\begin{lemma}
    \label{ct:lifting FDMR}
    For a given identifying vector $\mathbf{v}$ with special form in Definition \ref{def:identifying vector},
    let $\mathcal{M}$ be an $(q,k,n-k-\Delta,d_f)$ rank metric code
    with Ferrers diagram $\mathcal{F}$ corresponding to $\mathbf{v}$ constructed by Lemma \ref{ct:FDRM} or
    Lemma \ref{ct:FDRM improved}. We require that for $\textbf{M} \in \mathcal{M}_3$ rank metric code in construction satisfying
    $\operatorname{rank}(\textbf{M}) \leq u_1 - d_f$.
    By lifting the $\mathcal{M}$, $\mathcal{F}_c$ is an $(n,2d_f,k)_q$ CDC code such that for all codewords $U \in \mathcal{F}_c$, $i(U)=\mathbf{v}$.
    The cardinality satisfies

    \begin{align*}
        \# \mathcal{F}_c = \begin{cases}
            \Lambda_2 \cdot \Lambda_3 &  0 \leq \delta_1 - \Delta - u_1 < b_1, \\
            \min ( \# \mathcal{M}_1,  \# \mathcal{M}_2 ) \cdot \Lambda_3 &  b_1 \leq \delta_1 - \Delta - u_1 < d_f,\\
            s \cdot \Lambda_1 \cdot \Lambda_2 \cdot \Lambda_3  & \delta_1 - \Delta - u_1 \geq d_f,\\
        \end{cases}
    \end{align*}

    \noindent where $\Lambda_1 = m(q, u_1, \delta_1 - \Delta - u_1, d_f)$, $\Lambda_2 = m(q, u_2, \delta_2 - u_2, d_f)$,\\

    \noindent $\Lambda_3 = m(q, u_1, \delta_2 - u_2, d_f, u_1 - d_f),$
    $s = \min \left( \frac{m\left(q, u_1, \delta_1 - \Delta - u_1 , c_1\right)}{m\left(q, u_1,  \delta_1 - \Delta - u_1, d_f\right)},
    \frac{m\left(q, u_2, \delta_2 - u_2, c_2\right)}{m\left(q, u_2, \delta_2 - u_2, d_f\right)}
    \right)$.
\end{lemma}
\begin{proof}
    By construction, we have

    $$
        \mathcal{F}_c = \left\{ \mathrm{R} \begin{pmatrix}
            \mathbf{O}_1 & \mathbf{I}_{u_1} & \mathbf{M}_1 & \mathbf{O}_2 & \mathbf{M}_3 \\
            \mathbf{O}_3 & \mathbf{O}_4 & \mathbf{O}_5 & \mathbf{I}_{u_2} & \mathbf{M}_2
        \end{pmatrix}:  \begin{pmatrix}
            \mathbf{M}_1 & \mathbf{M}_3 \\
            \mathbf{O}_5 & \mathbf{M}_2
        \end{pmatrix} \in\mathcal{M} \right\},
    $$

    \noindent where $\operatorname{rank}(\mathbf{M}_3) \leq u_1 - d_f$,
    $\mathbf{O}_1 =\mathbf{O}_{u_1 \times \Delta},$ and $\mathbf{O}_i$ for $i=2,3,4,5$ are zero matrices of compatible size.\\

    For $W \in \mathcal{F}_c$,  $W$ is $k$-dimensional subspace in $\mathbf{F}_q^n$ since $\operatorname{rank}(\xi(W))=k.$
    Then let $W_1, W_2$ be two $k$-dimensional subspaces in $\mathcal{F}_c$, $W_1 \neq W_2$, we have
    $\operatorname{dis}(W_1, W_2) \geq 2\operatorname{d_R}(\mathcal{M}) \geq 2d_f$.
    The cardinality of $\mathcal{F}_c$ can be calculated from the cardinality of $\mathcal{M}$ given in Lemma \ref{ct:FDRM} and
    Lemma \ref{ct:FDRM improved}. \\
\end{proof}

\begin{proposition}
    \label{multilevel blocks}
    Let $C$ be a subspace code as in Theorem \ref{linkage construction} with $n = n_1 + n_2$ and $n_1 \geq k$, $n_2 \geq k$.
    Set $\delta_1 = n_1, \delta_2 = n_2, d_f = \frac{d}{2}$ and $H$ is a set consisting of vectors with special form as in Definition \ref{def:identifying vector}.
    For $v_j \neq v_{j^{\prime}} \in H $, we assume $\operatorname{d_h}(v_j, v_{j^{\prime}}) \geq d$.\\

    Then $\mathcal{L}_f = \bigcup_{j} \mathcal{L}_{j}$ is an $(n,d,k)_q$ CDC, where
    $\mathcal{L}_{j}$ is an $(n,d,k)_q$ CDC code lifted by $\mathcal{M}_j$ corresponding to $v_j$ for all $1 \leq j \leq \# H$ as in Lemma \ref{ct:lifting FDMR}.\\

    Moreover, $\mathcal{L} = \mathcal{L}_f \cup C$ is also an $(n,d,k)_q$ CDC.
\end{proposition}
\begin{proof}
    It is clearly that the elements of $\mathcal{L}$ are $k$-dimensional subspaces in $\mathbf{F}_q^n$ from Lemma \ref{ct:lifting FDMR}
    and Theorem \ref{linkage construction}.\\

    Let $W_1, W_2$ be two elements in $\mathcal{L}$. We analyse the following cases.\\

    (1) When $W_1 \in \mathcal{L}_j$ and $W_2 \in \mathcal{L}_{j^{\prime}}$ for $1 \leq j \leq j^{\prime} \leq \# H$,
    if $j=j^{\prime}$, $\operatorname{dis}(W_1,W_2) \geq d$ from Lemma \ref{ct:lifting FDMR}, if $j \neq j^{\prime}$, $\operatorname{dis}(W_1,W_2) \geq \operatorname{d_h}(i(W_1), i(W_2)) \geq d$ from Lemma \ref{lb:hamming lower bound}.\\

    (2) If $W_1 \in \mathcal{L}_j$ for $1 \leq j \leq \# H$ and $W_2 \in C$, we have
    $$
        W_1 = R\begin{pmatrix}
            \mathbf{O}_1 & \mathbf{I}_{u_1} & \mathbf{M}_1 & \mathbf{O}_2 & \mathbf{M}_3 \\
            \mathbf{O}_3 & \mathbf{O}_4 & \mathbf{O}_5 & \mathbf{I}_{u_2} & \mathbf{M}_2
        \end{pmatrix}
    $$ where
    $
    \begin{small}
        \begin{pmatrix}
            \mathbf{M}_1 & \mathbf{M}_3 \\
            \mathbf{O}_5 & \mathbf{M}_2
        \end{pmatrix}\end{small} \in \mathcal{M}_j, $
    $\operatorname{rank}(\mathbf{M}_3) \leq u_1 - \frac{d}{2}$,
    $\mathbf{O}_1 =\mathbf{O}_{u_1 \times \Delta}$ and $\mathbf{O}_i$ for\\

    \noindent $i=2,3,4,5$ are zero matrices of compatible sizes.
    With the same notations used in Lemma \ref{disjoint lemma}, we have
    \begin{align*}
        \operatorname{dim}(S_2 + W_1) &= n_1 + u_2 + \operatorname{rank}(\mathbf{M}_3) \leq n_1 + u_2 + u_1 - \frac{d}{2},\\
        \operatorname{dim}(S_1 + W_1) &= u_1 + n_2,
    \end{align*}
    then $\operatorname{dim}(S_2 \cap W_1) \geq \frac{d}{2}$ and
    $\operatorname{dim}(S_1 \cap W_1) = u_2 \geq \frac{d}{2}$.
    The Lemma \ref{inserting sufficient condition} gives that $\operatorname{dis}(W_1, W_2) \geq d$.\\ \\
\end{proof}

We consider $H$ with identifying vectors in the following two cases.\\

I. $H = \{ v_1, v_2 \}$. Set $\Delta=0, u_1 \geq d, u_2 \geq \frac{d}{2}, n_1 - u_1 \geq \frac{d}{2}, n_2 - u_2 \geq \frac{d}{2}$,
where
\begin{gather*}
    v_1 = (\underbrace{\overbrace{1 \cdots 1}^{u_1} 0 \cdots 0}_{n_1} \underbrace{\overbrace{1 \cdots 1}^{u_2} 0 \cdots 0}_{n_2}), \\
    v_2 = (\underbrace{\overbrace{1 \cdots 1}^{u_1 - \frac{d}{2}} 0 \cdots 0}_{n_1} \underbrace{\overbrace{1 \cdots 1}^{u_2 + \frac{d}{2}} 0 \cdots 0}_{n_2})
\end{gather*}
By construction in Proposition \ref{multilevel blocks}, we obtain the new $(n,d,k)_q$ CDC code with the lower bounds given in Corollary \ref{lb:multilevel blocks}.\\

For example $n=12, d=4, k=6$, $n_1 = n_2 = 6$, the identifying vectors with parameters
$u_1 = 4, u_2 = 2$ and $u_1^{\prime} = 2, u_2^{\prime} = 4$ are given by $v_1 = (111100 \enspace 110000)$, $v_2=(110000 \enspace 111100)$.
Since $n_i - u_i \geq \frac{d}{2}, n_i - u_i^{\prime} \geq \frac{d}{2}$ for $i=1,2$,
we consider CDCs which are lifted by rank metric code $\mathcal{M}$ constructed in Lemma \ref{ct:FDRM improved}.
When $q=2$,
the CDC $\mathcal{L}_{1}$ with the identifying vectors $v_1$ has cardinality of $\# \mathcal{L}_{1} = 2154496,$
the CDC $\mathcal{L}_{2}$ with the identifying vectors $v_2$ has cardinality of $\# \mathcal{L}_{2} = 4096$.
Thus we have
\begin{align*}
    A_2(12,4,6) & \geq \# C + \# \mathcal{L}_{1} + \# \mathcal{L}_{2} \\
    & \geq 1212418496 + 2154496 + 4096 = 1214577088.
\end{align*}
The new lower bounds for $A_q(12,4,6)$ and $A_q(18,6,9)$ for $q=2,3,4,5,7,8,9$ are given in Corollary \ref{lb: detail bounds}.
These new lower bounds improve the lower bounds in Theorem \ref{ct:multi-blocks} and are better than the lower bounds in \cite{table}.\\

When $\Delta=0$, $\delta_1 - u_1 \geq \frac{d}{2}$, $\delta_2 - u_2 \geq \frac{d}{2}$, the CDC lifted by the
rank metric code $\mathcal{M}$ constructed in Lemma \ref{ct:FDRM improved} is a special case of
\textit{block construction} in Proposition \ref{ct:two-blocks}.\\

II. $H = \{ v_1, v_2, \cdots, v_{\lambda} \}$ for $1 \leq \lambda \leq \lfloor \frac{n_1}{u_1} \rfloor$. Set $u_1 \geq \frac{d}{2}, u_2 \geq \frac{d}{2}$,
where
\begin{align*}
    v_1 & = (\underbrace{\overbrace{1 \cdots 1}^{u_1} 0 \cdots 0}_{n_1} \underbrace{\overbrace{1 \cdots 1}^{u_2} 0 \cdots 0}_{n_2}), \\
    v_2 & = (\underbrace{\overbrace{0 \cdots 0}^{u_1} \overbrace{1 \cdots 1}^{u_1} 0 \cdots 0}_{n_1} \underbrace{\overbrace{1 \cdots 1}^{u_2} 0 \cdots 0}_{n_2}), \\
    v_{\lambda} &= (\underbrace{\overbrace{0 \cdots 0}^{u_1} \overbrace{0 \cdots 0}^{u_1} \overbrace{1 \cdots 1}^{u_1} \overbrace{0 \cdots 0}^{n_1 - \lambda u_1} }_{n_1} \underbrace{\overbrace{1 \cdots 1}^{u_2} 0 \cdots 0}_{n_2})
\end{align*}
It is easy to check that for $1 \leq j < j^{\prime} \leq \# H$ and $v_j, v_{j^{\prime}} \in H$, $\operatorname{dis}(v_j, v_j^{\prime}) \geq 2u_1 \geq d.$
By construction in Proposition \ref{multilevel blocks}, we obtain the new $(n,d,k)_q$ CDC code with the lower bounds given in Corollary \ref{lb:multilevel blocks case2}.\\

For example $n=14, d=6, k=7$, $n_1 = n_2 = 7$, $u_1=3, u_2=4$, $c_1=2, c_2=1$, the identifying vectors are
given by $v_1 = (1110000 \enspace 1111000)$, $v_2=(0001110 \enspace 1111000)$ for $\lambda=\lfloor \frac{n_1}{u_1} \rfloor=2$.
We consider CDCs which are lifted by rank metric code $\mathcal{M}$ constructed in Lemma \ref{ct:FDRM} and Lemma \ref{ct:FDRM improved}.
When $q=2$,
the CDC $\mathcal{L}_{1}$ with identifying vector $v_1$ has cardinality of $\# \mathcal{L}_{1} =
s \cdot m(q,3,4,3) \cdot m(q,4,3,3) \cdot m(q,3,3,3,0) = 4096,$ where $s= \min\left( \frac{m(q,3,4,2)}{m(q,3,4,3)}, \frac{m(q,4,3,1)}{m(q,4,3,3)} \right)=q^4.$
The CDC $\mathcal{L}_{2}$ with identifying vector $v_2$ has cardinality of $\# \mathcal{L}_{2} =
m(q,4,3,3) \cdot m(q,3,3,3,0) = 16$.
Thus we have
\begin{align*}
    A_2(14,6,7) & \geq \# C + \# \mathcal{L}_{1} + \# \mathcal{L}_{2} \\
    & \geq 34532238024 + 4096 + 16 = 34532242136.
\end{align*}
The new lower bounds for $A_q(14,6,7)$ are given in Corollary \ref{lb: detail bounds}.\\



\section{New lower bounds}

For the rank metric codes needed in Proposition \ref{ct:two-blocks}, Theorem \ref{ct:multi-blocks} and Lemma \ref{ct:FDRM improved},
we follow the subcode construction in Lemma \ref{subcode construction} (or see Corollary 4.5 in \cite{CKMP}).
From the lower bounds in Theorem \ref{Delsarte} and Theorem \ref{linkage construction},
we have the following result in Theorem \ref{ct:multi-blocks}.\\

\begin{corollary}
    \label{lb: improved blocks construction}
    Let $n_1 + n_2 =n , a_1 + a_2 =k , b_1 + b_2  \geq \frac{d}{2}$ and $a_i \leq t_i \leq n_i - \frac{d}{2}, n_i \geq k, a_i \geq \frac{d}{2}, 1 \leq b_i \leq \frac{d}{2}$,
    for $i=1,2$.

    \begin{align*}
        \mathbf{A}_{q}(n, d, k) &\geq \mathbf{A}_{q}\left(n_{1}, d, k\right) \cdot m(q, k, n_{2}, \frac{d}{2}) + \Theta \cdot \mathbf{A}_{q}\left(n_{2}, d, k\right) \\
        & + s \cdot \left(\mathbf{A}_{q}\left(t_{1}, d, a_{1}\right) \cdot m(q, a_{1}, n_{1}-t_{1}, \frac{d}{2}) \cdot \Delta_{1} \right.\\
        & \left. \cdot \mathbf{A}_{q}\left(t_{2}, d, a_{2}\right)
        \cdot m (q, a_{2}, n_{2}-t_{2}, \frac{d}{2}) \cdot \Delta_{2} \right),\\
    \end{align*}

    \noindent where $\Theta = 1 + \sum_{u=\frac{d}{2}}^{k-\frac{d}{2}} r(q, k, n_{1}, \frac{d}{2}, u),$
    $s=\min \left( \frac{m(q, a_{1}, n_{1}-t_{1}, b_{1})}{m(q, a_{1}, n_{1}-t_{1}, \frac{d}{2})},
    \frac{m(q, a_{2}, n_{2}-t_{2}, b_{2})}{m(q, a_{2}, n_{2}-t_{2}, \frac{d}{2})} \right).$\\

    \noindent $\Delta_{1}=m\left(q, a_{1}, n_{2}-t_{2}, \frac{d}{2}, a_{1}-\frac{d}{2}\right),$
    $\Delta_{2}=m\left(q, a_{2}, n_{1}-t_{1}, \frac{d}{2}, a_{2}-\frac{d}{2}\right).$\\

\end{corollary}

If $n_1 - t_1 \geq a_1$ and $n_2 - t_2 \geq a_2$ , the improved lower bound is given by Theorem \ref{ct:parallel multiple blocks}.\\

\begin{corollary}
    \label{lb:intergration}
    Let $n_1 + n_2 =n , a_1 + a_2 =k , b_1 + b_2  \geq \frac{d}{2}, c_1 + c_2 \leq k - \frac{d}{2}$ and $a_i \leq t_i \leq n_i - a_i, n_i \geq k, a_i \geq \frac{d}{2}, 1 \leq b_i \leq \frac{d}{2}, b_i \leq c_i \leq a_i$,
    for $i=1,2$.
    Then

        \begin{align*}
                \mathbf{A}_{q}(n, d, k) &\geq \mathbf{A}_{q}\left(n_{1}, d, k\right) \cdot m(q, k, n_{2}, \frac{d}{2}) + \Theta \cdot \mathbf{A}_{q}\left(n_{2}, d, k\right) \\
                & +  s \cdot \left(\mathbf{A}_{q}\left(t_{1}, d, a_{1}\right) \cdot m(q, a_{1}, n_{1}-t_{1}, \frac{d}{2}) \cdot \Delta_{1} \right. \\
                & \left. \cdot \mathbf{A}_{q}\left(t_{2}, d, a_{2}\right) \cdot m(q, a_{2}, n_{2}-t_{2}, \frac{d}{2}) \cdot \Delta_{2}\right) \\
                & + \min( \Delta_{3}, \Delta_{4} ) \cdot \mathbf{A}_q(n_1 - t_1, d, a_1) \cdot \mathbf{A}_q(n_2 - t_2, d, a_2),\\
        \end{align*}

    \noindent where  $\Theta = 1+\sum_{u=\frac{d}{2}}^{k-\frac{d}{2}} r(q, k, n_{1}, \frac{d}{2}, u),$
    $s=\min \left( \frac{m(q, a_{1}, n_{1}-t_{1}, b_{1})}{m(q, a_{1}, n_{1}-t_{1}, \frac{d}{2})}, \frac{ m(q, a_{2}, n_{2}-t_{2}, b_{2})}{m(q, a_{2}, n_{2}-t_{2}, \frac{d}{2})} \right),$\\

    \noindent $\Delta_{1}=m\left(q, a_{1}, n_{2}-t_{2}, \frac{d}{2}, a_{1}-\frac{d}{2}\right),$
    $\Delta_{2}=m\left(q, a_{2}, n_{1}-t_{1}, \frac{d}{2}, a_{2}-\frac{d}{2}\right),$ \\

    \noindent $\Delta_{3}=m(q,a_1,t_1,b_1,c_1),$
    $\Delta_{4}=m(q,a_2,t_2,b_2,c_2).$\\
\end{corollary}

From Proposition \ref{multilevel blocks} we have the following result
which inserts the multilevel type construction CDC into linkage construction CDC for
the case of two identifying vectors.\\

\begin{corollary}
    \label{lb:multilevel blocks}
    let $n_1 + n_2 =n, n_1 \geq k, n_2 \geq k, u_1 + u_2 = k, u_1 \geq d, u_2 \geq \frac{d}{2}$ and $ c_1 + c_2 \geq \frac{d}{2}, 1 \leq c_i \leq \frac{d}{2}, i=1,2$,

    \begin{align*}
        \mathbf{A}_q(n,d,k) &\geq \mathbf{A}_{q}\left(n_{1}, d, k\right) \cdot m\left(q, k, n_{2}, \frac{d}{2}\right) + \Theta \cdot \mathbf{A}_{q}\left(n_{2}, d , k\right) \\
        & + s_1 \cdot m\left(q, u_1, n_1 - u_1, \frac{d}{2}\right) \cdot \Delta_{1} \cdot m\left(q, u_2, n_2 - u_2, \frac{d}{2}\right) \\
        & + s_2 \cdot m\left(q, u_1^{\prime}, n_1 - u_1^{\prime}, \frac{d}{2}\right) \cdot \Delta_{2} \cdot m\left(q, u_2^{\prime}, n_2 - u_2^{\prime}, \frac{d}{2}\right),\\
    \end{align*}

    \noindent where $\Theta = 1 + \sum_{u=\frac{d}{2}}^{k-\frac{d}{2}} r(q, k, n_{1}, \frac{d}{2}, u),$
    $u_1^{\prime} = u_1 - \frac{d}{2},$ $u_2^{\prime} = u_2 + \frac{d}{2},$ \\

    \noindent $\Delta_{1} = m(q,u_1, n_2 - u_2, \frac{d}{2}, u_1 - \frac{d}{2}),$
    $\Delta_{2} = m(q,u_1^{\prime}, n_2 - u_2^{\prime}, \frac{d}{2}, u_1^{\prime} - \frac{d}{2}),$ \\

    \noindent $s_1 = \min(\alpha_{i}: i=1,2)$, $\alpha_{i}= \frac{m \left(q, u_i, n_i - u_i, c_i \right)}{m(q, u_i, n_i - u_i, \frac{d}{2})}$, \\

    \noindent $s_2 = \min(\beta_{i}: i=1,2 )$, $\beta_{i} = \frac{m(q, u_i^{\prime}, n_i - u_i^{\prime}, c_i)}{m(q, u_i^{\prime}, n_i - u_i^{\prime}, \frac{d}{2})}.$\\

\end{corollary}

From Proposition \ref{multilevel blocks} we have the following result
which inserts the multilevel type construction CDC into linkage construction CDC for
the case of $\lambda = \lfloor \frac{n_1}{u_1} \rfloor$ identifying vectors.\\

\begin{corollary}
    \label{lb:multilevel blocks case2}
    let $n_1 + n_2 =n, n_1 \geq k, n_2 \geq k, u_1 + u_2 = k, u_1 \geq \frac{d}{2}, u_2 \geq \frac{d}{2}$ and $ b_1 + b_2 \geq \frac{d}{2}, 1 \leq b_i \leq \frac{d}{2}, i=1,2$,

        \begin{align*}
            \mathbf{A}_q(n,d,k) &\geq \mathbf{A}_{q}\left(n_{1}, d, k\right) \cdot m\left(q, k, n_{2}, \frac{d}{2}\right) \\
            & + \left(1+\sum_{u=\frac{d}{2}}^{k-\frac{d}{2}} r\left(q, k, n_{1}, \frac{d}{2}, u\right)\right) \cdot \mathbf{A}_{q}\left(n_{2}, d , k\right) + \sum\limits_{i = 1}^{\lambda} \mathcal{L}_i,\\
        \end{align*}

    \noindent where $\mathcal{L}_i = \begin{cases}
    \Lambda_1 \cdot \Lambda_2 &  0 \leq n_1 - i\cdot u_1 < b_1 \\
    \min ( \Lambda_3,  \Lambda_4 ) \cdot \Lambda_2 &  b_1 \leq n_1 - i \cdot u_1 < \frac{d}{2}\\
    s \cdot \Lambda_5 \cdot \Lambda_1 \cdot \Lambda_2  & n_1 - i \cdot u_1 \geq \frac{d}{2}
    \end{cases},$\\ \\

    \noindent $\Lambda_1 = m(q, u_2, n_2 - u_2, \frac{d}{2}),$
    $\Lambda_2 = m(q, u_1, n_2 - u_2, \frac{d}{2}, u_1 - \frac{d}{2}),$ \\

    \noindent $\Lambda_3 = m(q, u_1, n_1 - i \cdot u_1, b_1),$ $\Lambda_4 = m(q, u_2, n_2 - u_2, b_2),$ \\

    \noindent $\Lambda_5 = m(q, u_1, n_1 - i \cdot u_1, \frac{d}{2})$ and
    $s = \min \left( \frac{m\left(q, u_1, n_1 - i \cdot u_1 , b_1\right)}{m\left(q, u_1,  n_1 - i \cdot u_1, \frac{d}{2}\right)},
    \frac{m\left(q, u_2, n_2 - u_2, b_2\right)}{m\left(q, u_2, n_2 - u_2, \frac{d}{2}\right)}
    \right)$.\\ \\
\end{corollary}

\begin{corollary}
    \label{lb: detail bounds}
    We have the following lower bounds for constant dimension subspace codes with $d \leq k$.

    \allowdisplaybreaks
    \begin{align*}
            && \mathbf{A}_{q}(12, 4, 6) \geq \enspace &q^{30} + q^{26} + q^{25} + 2q^{24} + q^{23} + q^{22} - q^{21} - 2q^{20} - 3q^{19} \\
            && \ & - q^{18} - q^{17} + 3q^{15} + 3q^{14} + 4q^{13} + 4q^{12} + q^{11} \\
            && \ & - q^{10} - 3q^{9} - 3q^{8} - 2q^{7} - q^{6}.
            \\
            \\
            \\
            && \mathbf{A}_{q}(14,6,7) \geq \enspace &q^{35} + q^{26} + q^{25} + 2q^{24} + 3q^{23} + 3q^{22} + 2q^{21} + q^{20} - 2q^{19} \\
            && \ & - 5q^{18} - 8q^{17} - 11q^{16} - 11q^{15} - 10q^{14} - 7q^{13} - 3q^{12} \\
            && \ & + 2q^{11} + 5q^{10} + 8q^9 + 8q^8 + 9q^7 + 6q^6 + 5q^5 + 3q^4 + q^3.
            \\
            \\
            \\
            && \mathbf{A}_{q}(15, 4, 5) \geq \enspace  &q^{40} + \mathbf{A}_q(10,4,5)(q^{16} + q^{15} + 2q^{14} + q^{13} - 2q^{11} - 3q^{10} \\
            && \ & - 4q^{9} - 2q^{8}  + q^{6} + 3q^{5} + 2q^{4} + q^{3}) + \mathbf{A}_q(7,4,3)q^{12}.
            \\
            \\
            \\
            && \mathbf{A}_{q}(16, 6, 8) \geq \enspace &q^{48} + q^{39} + q^{38} + 2q^{37} + 3q^{36} + 3q^{35} + 3q^{34} + 2q^{33} - 4q^{31} \\
            && \ & - 6q^{30} - 10q^{29} - 10q^{28} - 11q^{27} - 7q^{26} - 3q^{25} + 6q^{24} \\
            && \ & + 12q^{23} + 19q^{22} + 23q^{21} + 25q^{20} + 22q^{19} + 16q^{18} + 9q^{17} \\
            && \ & - 7q^{15} - 13q^{14} - 15q^{13} - 17q^{12} - 13q^{11} - 11q^{10} - 8q^{9} \\
            && \ & - 5q^{8} - 4q^{7} - 2q^{6} + q^{4} + q^{3}.
            \\
            \\
            \\
            && \mathbf{A}_{q}(18, 4, 6) \geq \enspace &q^{60} + \mathbf{A}_q(12,4,6)(q^{26} + q^{25} + 2q^{24} + q^{23} + q^{22} - q^{21} \\
            && \ & - 3q^{20} - 4q^{19} - 3q^{18} - 2q^{17} + 4q^{15} + 5q^{14} + 5q^{13} + 3q^{12} \\
            && \ & + q^{11} - q^{10} - 3q^{9} - 3q^{8}- 2q^{7} - q^{6}) + \mathbf{A}_q(8,4,4)(q^{28} \\
            && \ & + q^{27} + 2q^{26} + q^{25} - q^{23} - 2q^{22} - q^{21}).
            \\
            \\
            \\
            && \mathbf{A}_{q}(18, 6, 6) \geq \enspace & \mathbf{A}_q(12,6,6)q^{24} + \mathbf{A}_q(6,6,3)q^{15} + (q^{21} + q^{20} + 2q^{19} \\
            && \ & + 3q^{18} + 3q^{17} + 3q^{16} + 3q^{15} + 2q^{14} + q^{13} + q^{12} - q^{9} \\
            && \ & - q^{8} - 2q^{7} - 3q^{6} - 3q^{5} - 3q^{4} - 3q^{3} - 2q^{2} - q).\\
            \\
            \\
            \\
            && \mathbf{A}_{q}(18,6,9) \geq \enspace & q^{63} + q^{54} + q^{53} + 2q^{52} + 3q^{51} + 3q^{50} + 3q^{49} + 3q^{48} + q^{47} \\
            && \ & - 2q^{46} - 5q^{45} - 9q^{44} - 11q^{43} - 13q^{42} - 12q^{41} - 10q^{40} \\
            && \ & - 3q^{39} + 3q^{38} + 12q^{37} + 18q^{36} + 24q^{35} + 24q^{34} + 23q^{33} \\
            && \ & + 15q^{32} + 6q^{31} - 7q^{30} - 19q^{29} - 29q^{28} - 37q^{27} - 39q^{26} \\
            && \ & - 39q^{25} - 31q^{24} - 22q^{23} - 8q^{22} + 2q^{21} + 14q^{20} + 20q^{19} \\
            && \ & + 27q^{18} + 24q^{17} + 23q^{16} + 17q^{15} + 14q^{14} + 8q^{13} + 5q^{12} \\
            && \ & + 2q^{11} + q^{10}. \\
    \end{align*}
\end{corollary}

\section{Conclusion}

After pioneering works in \cite{Silberstein1,EtzionVardy,Silberstein2,Silberstein3,Gluesing,Heinlein2} about the construction of constant dimension subspace codes,
new lower bounds from various constructions have been developed extensively in \cite{XuChen,CHWX,Li,CKMP,Heinlein1,LCF,Kurz} since 2018.
On the other hand there are still big gaps between presently best upper bounds and lower bounds for small parameters $n \leq 19$ and $q\leq 9$ in \cite{table}.
It seems that new constant dimension subspaces can be inserted into some most effective constructions.
In this paper we present two parameter-controlled inserting constructions from this idea.
Our constructions give highly non-trivial better lower bounds better than previous lower bounds.
$141$ new constant dimension subspace codes with distance $4,6,8$ for small parameters $n \leq 19$ and $q \leq 9$ are given in Table 1-9.\\

\newpage

\begin{longtable}{|l@{\extracolsep{\fill}}|l|l|l||}
    \caption{\label{tab:d=4}Theorem \ref{ct:multi-blocks} d=4}\\ \hline
    ${\bf A}_q(n,d,k)$ & New & Old\\ \hline  \hline \endfirsthead
    \multicolumn{4}{r}{continued table} \\ \hline
    ${\bf A}_q(n,d,k)$ & New & Old\\ \hline \endhead
    ${\bf A}_2(12,4,6)$ & 1214 5729 92 &\tabincell{l}{1212 4910 81}   \\ \hline
    ${\bf A}_3(12,4,6)$ &\tabincell{l}{2099 4929 7978 267}&\tabincell{l}{2099 4378 4809 333}   \\ \hline
    ${\bf A}_4(12,4,6)$ &\tabincell{l}{1159 1944 1176 9294 848}& \tabincell{l}{1159 1928 8551 2400 896}  \\ \hline
    ${\bf A}_5(12,4,6)$ &\tabincell{l}{9332 4349 9108 4302 96875}& \tabincell{l}{9332 4337 6349 6412 34375}   \\ \hline
    ${\bf A}_7(12,4,6)$ &\tabincell{l}{2255 0482 5265 0930 \\ 1245 6609 47}&\tabincell{l}{2255 0482 4318 3963\\ 0511 0436 89}  \\ \hline
    ${\bf A}_8(12,4,6)$ &\tabincell{l}{1238 2901 4650 6193 \\ 6330 6582 8352}&\tabincell{l}{1238 2901 4517 0956 \\ 0581 4349 0048}  \\ \hline
    ${\bf A}_9(12,4,6)$ &\tabincell{l}{4239 8506 4977 3534 \\ 7946 7906 80243}&\tabincell{l}{4239 8506 4839 1042 \\ 4395 4820 22091}   \\ \hline

    ${\bf A}_2(14,4,7)$ & 4980 1091 73760 & \tabincell{l}{4975 8590 33088}     \\ \hline
    ${\bf A}_3(14,4,7)$ &\tabincell{l}{1115 8069 9486 1976 27621}&\tabincell{l}{1115 7972 4707 5781 87435}   \\ \hline
    ${\bf A}_4(14,4,7)$ &\tabincell{l}{1944 8126 1068 5525 \\2931 117056}&\tabincell{l}{1944 8119 7073 7370 \\ 4380 940288}  \\ \hline
    ${\bf A}_5(14,4,7)$ &\tabincell{l}{2278 4276 5460 4718 \\1749 7177 734375 }& \tabincell{l}{2278 4275 9466 8902 \\ 8978 7216 796875}   \\ \hline
    ${\bf A}_7(14,4,7)$ &\tabincell{l}{3121 2771 2852 6308 \\ 1722 9550 5250 4601 2105}&\tabincell{l}{3121 2771 2665 4399 \\ 1534 2239 2501 1209 3363}  \\ \hline
    ${\bf A}_8(14,4,7)$ &\tabincell{l}{8509 4651 5464 2025 \\ 8527 8170 7789 6654 094336}&\tabincell{l}{8509 4651 5349 5070  \\ 8209 2409 6484 8376 872960}  \\ \hline
    ${\bf A}_9(14,4,7)$ &\tabincell{l}{1197 4590 5684 5506 \\ 9588 2882 6907 9025 \\ 2468 97611} &\tabincell{l}{1197 4590 5680 2122 \\ 9813 0213 9115 1254 \\ 3824 67123}   \\ \hline

    ${\bf A}_2(15,4,5)$ & 1252 4489 02208  & 1252 4485 86816   \\ \hline
    ${\bf A}_3(15,4,5)$ & \tabincell{l}{1239 9153 9128 2781 0424} & \tabincell{l}{1239 9153 9126 0619 9527} \\ \hline
    ${\bf A}_4(15,4,5)$ &\tabincell{l}{1215 5144 1280 2999 \\ 1645 38880} & \tabincell{l}{1215 5144 1280 2977 \\ 4883 75808}  \\ \hline
    ${\bf A}_5(15,4,5)$ &\tabincell{l}{9113 7155 3273 4825 \\2129 0267  1875} &\tabincell{l}{9113 7155 3273 4824 \\ 2924 9251 5625}   \\ \hline
    ${\bf A}_7(15,4,5)$ &\tabincell{l}{6369 9534 3303 4380 \\7169 4492 2899 996861}   &   \tabincell{l}{6369 9534 3303 4380 \\ 7166 7847 8121 377611}  \\ \hline
    ${\bf A}_8(15,4,5)$ &\tabincell{l}{1329 6039 3627 5552  \\ 4247 7739 2995 4536 \\ 32512}   &   \tabincell{l}{1329 6039 3627 5552 \\ 4247 7485 5872 3725 \\ 39392}  \\ \hline
    ${\bf A}_9(15,4,5)$ &\tabincell{l}{1478 3445 1659 2420 \\ 9511 5973 3248 7356 \\ 6503 455}   & \tabincell{l}{1478 3445 1659 2420  \\ 9511 5954 7743 3675 \\ 1358 413}   \\ \hline

    ${\bf A}_2(16,4,4)$ & 8059 6325 662 & 8059 6320 222 \\ \hline

    ${\bf A}_2(16,4,5)$ & 2002 1892 886936 & 2002 1891 625368 \\ \hline
    ${\bf A}_3(16,4,5)$ & 1004 3083 9766 0450 578410 & 1004 3083 9765 8456 080337  \\ \hline
    ${\bf A}_4(16,4,5)$ &\tabincell{l}{3111 7130 9429 0302 \\ 0873 3618 688}& \tabincell{l}{3111 7130 9429 0298 \\ 6191 5009 536}  \\ \hline
    ${\bf A}_5(16,4,5)$ &\tabincell{l}{5696 0714 9291 8139 \\ 1221 4040 9840 875} & \tabincell{l}{5696 0714 9291 8139 \\ 0991 3015 5934 625}   \\ \hline
    ${\bf A}_7(16,4,5)$ &\tabincell{l}{1529 4258 1299 7087 \\ 2784 6439 2254 8588 \\ 487731} & \tabincell{l}{1529 4258 1299 7087 \\ 2784  6308 6675 4436 \\ 144481}  \\ \hline
    ${\bf A}_8(16,4,5)$ &\tabincell{l}{5446 0577 1721 3633 \\ 0719 7615 9692 0231 \\ 0769 4592}& \tabincell{l}{5446 0577 1721 3633 \\ 0719  7599 7316 1459 \\ 1773 4912} \\ \hline
    ${\bf A}_9(16,4,5)$ &\tabincell{l}{9699 4183 7024 2962 \\ 7092 2774 1522 1892\\ 3118 366723} &\tabincell{l}{9699 4183 7024 2962 \\  7092  2772 6496 2544 \\ 1091 618321}  \\ \hline

    ${\bf A}_2(16,4,8)$ & 8168 0045 6478 22848 & 8160 5776 6327 40149  \\ \hline
    ${\bf A}_3(16,4,8)$ & \tabincell{l}{5336 9600 6404 7301 \\ 6301 5741 843} & \tabincell{l}{5336 9353 1575 0209 \\ 5137 0793 043} \\ \hline
    ${\bf A}_4(16,4,8)$ &\tabincell{l}{5220 5715 4021 0932 \\ 7828 4022 5766 309888} &    \tabincell{l}{5220 5709 4302 9781 \\ 8840 3820 4796 960768}  \\ \hline
    ${\bf A}_5(16,4,8)$ &\tabincell{l}{1390 6420 2069 8714 \\ 8586 1618 0634 7351 0742 1875} &  \tabincell{l}{1390 6420 1127 5456 \\ 2178 5003 6835 9069 8242 1875}   \\ \hline
    ${\bf A}_7(16,4,8)$ &\tabincell{l}{2116 9221 7279 9490 \\ 0559 0051 3249 1560 \\ 2872 8216 9823 9203}&   \tabincell{l}{2116 9221 7258 4282 \\ 5319 3267 9190  6528 \\ 3985 7934 6055 9203}  \\ \hline
    ${\bf A}_8(16,4,8)$ &\tabincell{l}{3742 5023 5695 3339 \\ 0948 5184 1673 2113 \\ 8071 8757 7771 1833088}&  \tabincell{l}{3742 5023 5688 0309 \\ 3259 7524 6537 6356 \\ 2060 4406 9619 3548288}  \\ \hline
    ${\bf A}_9(16,4,8)$ &\tabincell{l}{2739 4022 3246 1542 \\ 0667 5302 9314 3680  \\ 4036 6083 7004 6051 871523}&\tabincell{l}{2739 4022 3244 9001 \\ 7516 2201 2431 4566 \\ 4311 8182 5520 1568 894883}\\\hline

    ${\bf A}_2(17,4,5)$ & 3203 6595 7408552 & 3203 6594 9667112  \\ \hline
    ${\bf A}_3(17,4,5)$ & 8134 8354 0402 8193 8373 636 & 8134 8354 0402 1799 6409 822  \\ \hline
    ${\bf A}_4(17,4,5)$ &\tabincell{l}{7965 9830 8515 8475 \\ 5714 3479 63840} & \tabincell{l}{7965 9830 8515 8472 \\ 6337 2395 58720}  \\ \hline
    ${\bf A}_5(17,4,5)$ &\tabincell{l}{3560 0445 9366 3614 \\ 2164 1799 4533 338500} &  \tabincell{l}{3560 0445 9366 3614 \\ 2128 0741 2000 712875} \\ \hline
    ${\bf A}_7(17,4,5)$ &\tabincell{l}{3672 1513 7485 3733 \\ 9070 3590 3605 3137 \\ 2645 46455}& \tabincell{l}{3672 1513 7485 3733 \\ 9070 3537 7403 5429 \\ 7833 36374} \\ \hline
    ${\bf A}_8(17,4,5)$ &\tabincell{l}{2230 7052 4067 5218 \\ 2674 9141 5501 7163 \\  4469 3581 9264}&  \tabincell{l}{2230 7052 4067 5218 \\ 2674 9140 5936 7518 \\ 6961 8703 1040}  \\ \hline
    ${\bf A}_9(17,4,5)$& \tabincell{l}{6363 7883 9248 8962 \\ 4285 1438 5840 6278 \\ 7421 4317 118244}&\tabincell{l}{6363 7883 9248 8962 \\ 4285 1438 4600 8757 \\ 1564 9107 710443}\\\hline

    ${\bf A}_2(18,4,5)$ & 5125 9206 2259 6904 & 5125 9205 9163 1144    \\ \hline
    ${\bf A}_3(18,4,5)$ & 6589 1997 5982 6179 3869 53990 & 6588 8606 4307 3901 6378 89182   \\ \hline
    ${\bf A}_4(18,4,5)$ &\tabincell{l}{2039 2915 1385 7915 \\ 5716 6182 2704 5888} &    \tabincell{l}{2039 2822 3978 3579 \\ 7265 2526 6919 3216}    \\ \hline
    ${\bf A}_5(18,4,5)$ &\tabincell{l}{2225 0278 5986 2960 \\ 9839 3994 3854 0551 29500} & \tabincell{l}{2225 0264 9575 8734 \\ 6352 6612 2408 3590 25000}  \\ \hline
    ${\bf A}_7(18,4,5)$ &\tabincell{l}{8816 8354 5028 5580 \\ 4527 8570 8971 4171 \\ 0541 1823 9611}&\tabincell{l}{8816 8351 9743 4932 \\ 7793 7100 0707 9716 \\ 4432 5109 9490}\\ \hline
    ${\bf A}_8(18,4,5)$ &\tabincell{l}{9136 9686 6565 4404 \\ 7861  6176 8715 7466 \\ 7878 8938 8430848} & \tabincell{l}{9136 9685 8810 1376 \\ 1933 1999 5940 0530 \\ 0227 3383 7074432}    \\  \hline
    ${\bf A}_9(18,4,5)$ &\tabincell{l}{4175 2815 6429 5427 \\ 4523 1912 0967 5971 \\ 7973 6550 1634 012230} & \tabincell{l}{4175 2815 5218 0051 \\ 4325 7793 3240 9757 \\ 9235 7748 3683 264556}    \\ \hline

    ${\bf A}_2(18,4,6)$ & 1321 0683 8054 5845184 & 1321 0657 3684 4576704   \\ \hline
    ${\bf A}_3(18,4,6)$ &\tabincell{l}{4324 1984 5318 8854 \\ 8932 5355 54684}  &\tabincell{l}{4324 1984 5121 9278 \\ 9981 1406 81783}  \\ \hline
    ${\bf A}_4(18,4,6)$ &\tabincell{l}{1336 4977 3466 9987 \\ 4494 2103 7830 6137 62048}   &    \tabincell{l}{1336 4977 3466 8298 \\ 4303 3566 4941 2839 34208}  \\ \hline
    ${\bf A}_5(18,4,6)$ &\tabincell{l}{8691 5431 3455 6232 \\ 1645 6979 0404 3228 \\ 6621 093750}   &   \tabincell{l}{8691 5431 3455 6114 \\ 8125 0767 1449 6643 \\ 5791 015625}   \\ \hline
    ${\bf A}_7(18,4,6)$ &\tabincell{l}{5082 7312 1397 7132 \\ 1237 0271 5485 2231 \\ 8776 4110 1133 5753136}   &  \tabincell{l}{5082 7312 1397 7132 \\ 0481 2175 3530 1491 \\ 9036 5403 9631 9478635}  \\ \hline
    ${\bf A}_8(18,4,6)$ &\tabincell{l}{1532 9290 7353 3720 \\ 1342 0131 7476 1821\\ 4131 0078 6371 2572 1391104}  &   \tabincell{l}{1532 9290 7353 3720 \\ 1326 6141 0675 9748 \\ 0056 0957 6260 2317 8928128}  \\ \hline
    ${\bf A}_9(18,4,6)$ &\tabincell{l}{1797 3218 5298 8389 \\ 5320 8891 5801 6612 \\ 7145 1862 6711 1609 3928 898138}&     \tabincell{l}{1797 3218 5298 8389 \\ 5319 2080 5010 3864 \\ 5852 2272 4684 1493 0785 473613}   \\ \hline

    ${\bf A}_2(18,4,9)$ & 5353 1244 5248 1263 206400 & 5350 7797 0493 6727 838720 \\ \hline
    ${\bf A}_3(18,4,9)$ &\tabincell{l}{2297 3952 1671 4333 \\ 7216 5373 5752 8684349}  &\tabincell{l}{2297 3916 8156 5204 \\ 1702 9355 6329 2591869}  \\ \hline
    ${\bf A}_4(18,4,9)$ &\tabincell{l}{2242 2188 6155 0678 \\ 0283 4073 8745 2105 \\ 7495 0996 3776}   &    \tabincell{l}{2242 2187 9749 4966 \\ 5202 5921 0903 2573 \\ 1520 7125 4016}  \\ \hline
    ${\bf A}_5(18,4,9)$ &\tabincell{l}{2121 9513 6554 7144 \\ 5242 7277 8843 8380 \\ 9190 2732 8491 2109375}   &   \tabincell{l}{2121 9513 6267 2348 \\ 6969 6158 8067 3111 \\ 6395 4734 8022 4609375}   \\ \hline
    ${\bf A}_7(18,4,9)$ &\tabincell{l}{7035 1527 6242 6758 \\ 0416 2560 5506 1526 \\ 5952 2984 6057 3579\\ 8134 8390 86745}   &  \tabincell{l}{7035 1527 6232 4593 \\ 1444 0919 3941 7657 \\ 2694 5128 9708 2189 \\ 8343 1173 45945}  \\ \hline
    ${\bf A}_8(18,4,9)$ &\tabincell{l}{1053 4207 6388 2646 \\ 4848 8236 8607 3137 \\ 9760 8414 6989 8696 \\ 2431 7692 4558 393344}  &   \tabincell{l}{1053 4207 6388 0077 \\ 0579 2562 6795 4570 \\ 4178 5576 6160 5228 \\ 1449 8418 7966 849024}  \\ \hline
    ${\bf A}_9(18,4,9)$ &\tabincell{l}{5076 1676 4229 5809 \\ 7114 5598 7655 7785 \\ 3905 5475 7494 2459 \\ 6609 0855 5802 1960 31451}&     \tabincell{l}{5076 1676 4229 3227 \\ 8160 2958 9753 8502 \\ 2100 0195 2897 8613 \\ 0451 1752 6000 3851 92731}   \\ \hline

    ${\bf A}_2(19,4,5)$ & 8201 4791 1159 59488 & 8201 4790 9849 28448 \\ \hline
    ${\bf A}_3(19,4,5)$ &\tabincell{l}{5336 9771 2296 4435 \\ 3688 0971 278}  &\tabincell{l}{5336 9771 2296 4420 \\ 8282 2788 731}  \\ \hline
    ${\bf A}_4(19,4,5)$ &\tabincell{l}{5220 5625 3384 9331 \\ 8374 9756 9957 708800}   &    \tabincell{l}{5220 5625 3384 9331 \\ 8360 7699 4945 156096}  \\ \hline
    ${\bf A}_5(19,4,5)$ &\tabincell{l}{1390 6415 5984 9215 \\ 8116 1871 0095 6619 \\ 6694 0625}   &   \tabincell{l}{1390 6415 5984 9215 \\ 8116 1835 0560 2479 \\ 0444 0625}   \\ \hline
    ${\bf A}_7(19,4,5)$ &\tabincell{l}{2116 9221 3090 4127 \\ 3671 1759 7474 1768 \\ 4343 1973 3915 7355}   &  \tabincell{l}{2116 9221 3090 4127 \\ 3671 1759 7458 8168 \\ 3177 5206 3961 5255}  \\ \hline
    ${\bf A}_8(19,4,5)$ &\tabincell{l}{3742 5023 3368 6323 \\ 6894 5845 8213 9739 \\ 1254 9369 8700 5648896}  &   \tabincell{l}{3742 5023 3368 6323 \\ 6894 5845 8213 5482 \\ 5392 6608 2829 8182656}  \\ \hline
    ${\bf A}_9(19,4,5)$ &\tabincell{l}{2739 4022 2638 5331 \\ 7449 5366 8303 7639 \\ 3896 7183 5642 4596 423365}&     \tabincell{l}{2739 4022 2638 5331 \\ 7449 5366 8303 7559 \\ 5356 9335 6715 8197 694861}   \\ \hline

    ${\bf A}_2(19,4,6)$ & 4224 2622 2853 8904 2880 & 4224 2601 1357 7889 5040 \\ \hline
    ${\bf A}_3(19,4,6)$ &\tabincell{l}{1050 7720 6904 7370 \\ 9000 6492 9154 5617}  &\tabincell{l}{1050 7720 6899 4192 \\ 3508 9616 2997 7290}  \\ \hline
    ${\bf A}_4(19,4,6)$ &\tabincell{l}{1368 5732 6052 4843 \\ 9975 3210 8659 5253 \\  6240 1280}   &    \tabincell{l}{1368 5732 6052 4735 \\ 9003 1064 4794 6082\\ 5341 9520}  \\ \hline
    ${\bf A}_5(19,4,6)$ &\tabincell{l}{2716 1071 6124 6622 \\ 4559 3611 2536 3003 \\ 2935 1659 828125}   &   \tabincell{l}{2716 1071 6124 6620 \\ 9890 3533 6049 3670 \\ 9799 7900 062500}   \\ \hline
    ${\bf A}_7(19,4,6)$ &\tabincell{l}{8542 5463 4632 2734 \\ 0885 0727 8319 1374 \\ 6697 4647 3297 9651 9336064}   &  \tabincell{l}{8542 5463 4632 2734 \\ 0859 1485 1324 0784 \\ 2886 3871 7088 4593 7182221}  \\ \hline
    ${\bf A}_8(19,4,6)$ &\tabincell{l}{5023 1019 8748 9776 \\ 9686 7148 8089 4365 \\ 5262 3277 5608 9163 \\ 8292 3358208}  &   \tabincell{l}{5023 1019 8748 9776 \\ 9685 9264 4861 2659 \\ 3677 6922 5651 2318 \\ 8118 2314496}  \\ \hline
    ${\bf A}_9(19,4,6)$ &\tabincell{l}{1061 3005 8093 3177 \\ 5107 6560 0915 0855 \\ 9074 2028 9616 4433 \\ 9549 4258 9146667}&     \tabincell{l}{1061 3005 8093 3177 \\ 5107 6437 5387 4171 \\  9735 5772 2490 6664 \\ 8070 1103 2667942}   \\ \hline
\end{longtable}

\newpage

\begin{longtable}{|l@{\extracolsep{\fill}}|l|l|l||}
    \caption{\label{tab:d=6}Theorem \ref{ct:multi-blocks} d=6}\\ \hline
    ${\bf A}_q(n,d,k)$ & New & Old\\ \hline  \hline \endfirsthead
    \multicolumn{4}{r}{continued table} \\ \hline
    ${\bf A}_q(n,d,k)$ & New & Old\\ \hline \endhead

    ${\bf A}_2(18,6,6)$ & 2829 5832 3493518 & 2829 5832 3460750 \\ \hline
    ${\bf A}_3(18,6,6)$ &\tabincell{l}{7977 3414 6743 2777 8613776}  &\tabincell{l}{7977 3414 6743 2776 4264869}  \\ \hline
    ${\bf A}_4(18,6,6)$ &\tabincell{l}{7922 8596 9086 1399 \\ 5335 4256 05660}   &    \tabincell{l}{7922 8596 9086 1399 \\ 5334 3518 63836}  \\ \hline
    ${\bf A}_5(18,6,6)$ &\tabincell{l}{3552 7160 6160 5390 \\ 6089 3919 2136 113320}   &   \tabincell{l}{3552 7160 6160 5390 \\ 6089 3916 1618 535195}   \\ \hline
    ${\bf A}_7(18,6,6)$ &\tabincell{l}{3670 3369 3031 7493 \\ 1772 0054 8142 8975 \\ 2151 02688}   &  \tabincell{l}{3670 3369 3031 7493 \\ 1772 0054 8142 4227 \\ 6535 92745}  \\ \hline
    ${\bf A}_8(18,6,6)$ &\tabincell{l}{2230 0745 3917 5803 \\ 6632 4706 9642 8537 \\ 0816 5403 6856}  &   \tabincell{l}{2230 0745 3917 5803 \\ 6632 4706 9642 8501 \\ 8972 8194 8024}  \\ \hline
    ${\bf A}_9(18,6,6)$ &\tabincell{l}{6362 6854 5986 5481 \\ 8930 7460 2002 4319 \\ 4539 0102 248520}&     \tabincell{l}{6362 6854 5986 5481 \\ 8930 7460 2002 4317 \\ 3949 8970 153871}   \\ \hline
\end{longtable}

\newpage
\begin{longtable}{|l@{\extracolsep{\fill}}|l|l|l||}
    \caption{\label{tab:parallel multiple blocks d6}Theorem \ref{ct:parallel multiple blocks} d=6}\\ \hline
    ${\bf A}_q(n,d,k)$ & New & Old\\ \hline  \hline \endfirsthead
    \multicolumn{4}{r}{continued table} \\ \hline
    ${\bf A}_q(n,d,k)$ & New & Old\\ \hline \endhead
    ${\bf A}_2(12,6,6)$ & 1686 5664 & 1686 5630 \\ \hline
    ${\bf A}_3(12,6,6)$ &\tabincell{l}{2824 5422 1144}  &\tabincell{l}{2824 5422 0859}   \\ \hline
    ${\bf A}_4(12,6,6)$ &\tabincell{l}{2814 7651 9990 600}  &  \tabincell{l}{2814 7651 9989 404}  \\ \hline
    ${\bf A}_5(12,6,6)$ &\tabincell{l}{5960 4684 7522 26540}   & \tabincell{l}{5960 4684 7522 22945}   \\ \hline
    ${\bf A}_7(12,6,6)$ &\tabincell{l}{1915 8123 7048 5580 13104}   &  \tabincell{l}{1915 8123 7048 5579 94295}  \\ \hline
    ${\bf A}_8(12,6,6)$ &\tabincell{l}{4722 3665 2378 7141 634864}   &     \tabincell{l}{4722 3665 2378 7141 598584}  \\ \hline
    ${\bf A}_9(12,6,6)$ &\tabincell{l}{7976 6443 3116 7725 7540 500}   &   \tabincell{l}{7976 6443 3116 7725 7475 709}   \\ \hline

    ${\bf A}_2(16,6,8)$ & 2829 2768 4887 704 & 2829 2768 4884 928 \\ \hline
    ${\bf A}_3(16,6,8)$ &\tabincell{l}{7977 3403 8582 1485 4319 604}  &\tabincell{l}{7977 3403 8582 1485 4088 403}   \\ \hline
    ${\bf A}_4(16,6,8)$ &\tabincell{l}{7922 8596 7952 0959 8385 \\ 5275 78944}  &  \tabincell{l}{7922 8596 7952 0959 8385 \\ 5223 72608}  \\ \hline
    ${\bf A}_5(16,6,8)$ &\tabincell{l}{3552 7160 6144 6350 4786 \\ 5950 4366 844500}   & \tabincell{l}{3552 7160 6144 6350 4786 \\ 5950 4308 421875}   \\ \hline
    ${\bf A}_7(16,6,8)$ &\tabincell{l}{3670 3369 3031 6550 6402 \\ 6817 0441 9419 8449 80004}   &  \tabincell{l}{3670 3369 3031 6550 6402 \\ 6817 0441 9417 5869 40003}  \\ \hline
    ${\bf A}_8(16,6,8)$ &\tabincell{l}{2230 0745 3917 5728 7672 \\ 3615 6375 2919 8570 8555 1104}   &     \tabincell{l}{2230 0745 3917 5728 7672 \\ 3615 6375 2919 8474 2632 6528}  \\ \hline
    ${\bf A}_9(16,6,8)$ &\tabincell{l}{6362 6854 5986 5446 2048 \\ 6152 6050 7124 2490 8815 104484}   &   \tabincell{l}{6362 6854 5986 5446 2048 \\ 6152 6050 7124 2487 3957 350563}   \\ \hline

    ${\bf A}_2(19,6,6)$ & 4527 3330 8759 0608 & 4527 3330 8758 6958\\ \hline
    ${\bf A}_3(19,6,6)$ &\tabincell{l}{6461 6465 8861 9087 5700 28526}  &\tabincell{l}{6461 6465 8861 9087 5697 71903}   \\ \hline
    ${\bf A}_4(19,6,6)$ &\tabincell{l}{2028 2520 8086 0518 1180 \\ 3566 2610 8488}  &  \tabincell{l}{2028 2520 8086 0518 1180 \\ 3566 2059 9324}  \\ \hline
    ${\bf A}_5(19,6,6)$ &\tabincell{l}{2220 4475 3850 3369 1301 \\ 9528 8648 9837 85290}   & \tabincell{l}{2220 4475 3850 3369 1301 \\  9528 8648 9232 22695}   \\ \hline
    ${\bf A}_7(19,6,6)$ &\tabincell{l}{8812 4789 6969 2301 1184 \\ 5835 5714 5587 3536 6074 7838}   &  \tabincell{l}{8812 4789 6969 2301 1184 \\ 5835 5714 5587 3513 6047 4303}  \\ \hline
    ${\bf A}_8(19,6,6)$ &\tabincell{l}{9134 3853 1246 4091 8046 \\ 5999 2858 2455 7658 7389 \\ 6868656}   &     \tabincell{l}{9134 3853 1246 4091 8046 \\ 5999 2858 2455 7658 6409 \\ 8711928}  \\ \hline
    ${\bf A}_9(19,6,6)$ &\tabincell{l}{4174 5579 3021 7742 6700 \\ 4624 6291 8490 3548 8843 \\ 5323 855978}   &   \tabincell{l}{4174 5579 3021 7742 6700 \\ 4624 6291 8490 3548 8840 \\ 0068 059399}   \\ \hline

\end{longtable}

\begin{longtable}{|l@{\extracolsep{\fill}}|l|l|l||}
    \caption{\label{tab:parallel multiple blocks d8}Theorem \ref{ct:parallel multiple blocks} d=8}\\ \hline
    ${\bf A}_q(n,d,k)$ & New & Old\\ \hline  \hline \endfirsthead
    ${\bf A}_2(16,8,8)$ & 1099 5628 94524 & 1099 5628 93998\\ \hline
    ${\bf A}_3(16,8,8)$ &\tabincell{l}{1215 7665 9570 9072 2244}  &\tabincell{l}{1215 7665 9570 9071 1843}   \\ \hline
    ${\bf A}_4(16,8,8)$ &\tabincell{l}{1208 9258 2002 2366 9131 73944}  &  \tabincell{l}{1208 9258 2002 2366 9130 82908}  \\ \hline
    ${\bf A}_5(16,8,8)$ &\tabincell{l}{9094 9470 1780 7612 5205 9084 3140}   & \tabincell{l}{9094 9470 1780 7612 5205 9034 0195}   \\ \hline
    ${\bf A}_7(16,8,8)$ &\tabincell{l}{6366 8057 6090 9256 9002 \\ 8088 9458 243204}   &  \tabincell{l}{6366 8057 6090 9256 9002 \\ 8088 9451 403203}  \\ \hline
    ${\bf A}_8(16,8,8)$ &\tabincell{l}{1329 2279 9578 4921 3674 \\ 3970 1812 8694 41264}   &     \tabincell{l}{1329 2279 9578 4921 3674 \\ 3970 1812 8500 10488}  \\ \hline
    ${\bf A}_9(16,8,8)$ &\tabincell{l}{1478 0882 9414 3460 1431 \\ 1989 3459 8896 6906 884}   &   \tabincell{l}{1478 0882 9414 3460 1431 \\ 1989 3459 8891 7956 163}   \\ \hline
\end{longtable}

\newpage
\begin{longtable}{|l@{\extracolsep{\fill}}|l|l|l||}
    \caption{Multilevel type inserting construction I for d=4} \\ \hline
    \label{tab:mutilevel type linkage d4}
    ${\bf A}_q(n,d,k)$ & New & Old\\ \hline  \hline \endfirsthead
    \multicolumn{4}{r}{continued table} \\ \hline
    ${\bf A}_q(n,d,k)$ & New & Old\\ \hline \endhead
    ${\bf A}_2(12,4,6)$ & 1214 577088 & 1212 491081\\ \hline
    ${\bf A}_3(12,4,6)$ &\tabincell{l}{2099 4929 8509 708}  &\tabincell{l}{2099 4378 4809 333}   \\ \hline
    ${\bf A}_4(12,4,6)$ &\tabincell{l}{1159 1944 1178 6072 064}  &  \tabincell{l}{1159 1928 855 1240 0896}  \\ \hline
    ${\bf A}_5(12,4,6)$ &\tabincell{l}{9332 4349 9108 6744 37500}   & \tabincell{l}{9332 4337 6349 6412 34375}   \\ \hline
    ${\bf A}_7(12,4,6)$ &\tabincell{l}{2255 0482 5265 0931 5086 \\ 948148}   &  \tabincell{l}{2255 0482 4318 3963 0511 \\ 043689}  \\ \hline
    ${\bf A}_8(12,4,6)$ &\tabincell{l}{1238 2901 4650 6193 7017 \\ 8530 5088}   &     \tabincell{l}{1238 2901 4517 0956 0581 \\ 4349 0048}  \\ \hline
    ${\bf A}_9(12,4,6)$ &\tabincell{l}{4239 8506 4977 3534 8229 \\ 2202 16724}   &   \tabincell{l}{4239 8506 4839 1042 4395 \\ 4820 22091}   \\ \hline

    ${\bf A}_2(14,4,7)$ & 4980 1102 22336 & 4975 8590 33088 \\ \hline
    ${\bf A}_3(14,4,7)$ &\tabincell{l}{1115 8069 9489 6844 12022}  &\tabincell{l}{1115 7972 4707 5781 87435}   \\ \hline
    ${\bf A}_4(14,4,7)$ &\tabincell{l}{1944 8126 1068 5635 2442 744832}  &  \tabincell{l}{1944 8119 7073 7370 4380 940288}  \\ \hline
    ${\bf A}_5(14,4,7)$ &\tabincell{l}{2278 4276 5460 4719 1286 \\ 4609 375000}   & \tabincell{l}{2278 4275 9466 8902 8978 \\ 7216 796875}   \\ \hline
    ${\bf A}_7(14,4,7)$ &\tabincell{l}{3121 2771 2852 6308 1730 \\ 9342 7913 4362 4106}   &  \tabincell{l}{3121 2771 2665 4399 1534 \\ 2239 2501 1209 3363}  \\ \hline
    ${\bf A}_8(14,4,7)$ &\tabincell{l}{8509 4651 5464 2025 8528 \\ 9699 9940 1260 941312}   &     \tabincell{l}{8509 4651 5349 5070 8209 \\ 2409 6484 8376 872960}  \\ \hline
    ${\bf A}_9(14,4,7)$ &\tabincell{l}{1197 4590 5684 5506 9588 \\ 3004 2674 4484 3038 26412}   &   \tabincell{l}{1197 4590 5680 2122 9813 \\ 0213 9115 1254 3824 67123}   \\ \hline

    ${\bf A}_2(16,4,8)$ & 8164 2270 4541 53216 & 8160 5776 6327 40149 \\ \hline
    ${\bf A}_3(16,4,8)$ &\tabincell{l}{5336 9510 1888 0770 0238 8592396}  &\tabincell{l}{5336 9353 1575 0209 5137 0793043}   \\ \hline
    ${\bf A}_4(16,4,8)$ &\tabincell{l}{5220 5713 7400 2282 1019 \\ 4237 1362 013184}  &  \tabincell{l}{5220 5709 4302 9781 8840 \\ 3820 4796 960768}  \\ \hline
    ${\bf A}_5(16,4,8)$ &\tabincell{l}{1390 6420 1860 2751 9606 \\ 7703 5069 7021 4843 7500}   & \tabincell{l}{1390 6420 1127 5456 2178 \\ 5003 6835 9069 8242 1875}   \\ \hline
    ${\bf A}_7(16,4,8)$ &\tabincell{l}{2116 9221 7276 5722 7048 \\ 6063 6213 0694 8430 1055 \\ 3072 2356}   &  \tabincell{l}{2116 9221 7258 4282 5319 \\ 3267 9190 6528 3985 7934 \\ 6055 9203}  \\ \hline
    ${\bf A}_8(16,4,8)$ &\tabincell{l}{3742 5023 5694 3378 7991 \\ 6013 9376 6437 1325 8017 \\ 7902 6722816}   &     \tabincell{l}{3742 5023 5688 0309 3259 \\ 7524 6537 6356 2060 4406 \\ 9619 3548288}  \\ \hline
    ${\bf A}_9(16,4,8)$ &\tabincell{l}{2739 4022 3246 0031 0189 \\ 1833 9047 3260 4437 6692 \\ 4595 0047 702164}   &   \tabincell{l}{2739 4022 3244 9001 7516 \\ 2201 2431 4566 4311 8182 \\ 5520 1568 894883}   \\ \hline

    ${\bf A}_2(18,4,6)$ & 1321 0657 4623 0904768 & 1321 0657 3684 4576704 \\ \hline
    ${\bf A}_3(18,4,6)$ &\tabincell{l}{4324 1984 5121 9576 3914 \\ 1825 83632}  &\tabincell{l}{4324 1984 5121 9278 9981 \\ 1406 81783}   \\ \hline
    ${\bf A}_4(18,4,6)$ &\tabincell{l}{1336 4977 3466 8298 4560 \\ 6020 1548 6904 89344}  &  \tabincell{l}{1336 4977 3466 8298 4303 \\ 3566 4941 2839 34208}  \\ \hline
    ${\bf A}_5(18,4,6)$ &\tabincell{l}{8691 5431 3455 6114 8128 \\ 0792 9817 0120 7519 531250}   & \tabincell{l}{8691 5431 3455 6114 8125 \\ 0767 1449 6643 5791 015625}   \\ \hline
    ${\bf A}_7(18,4,6)$ &\tabincell{l}{5082 7312 1397 7132 0481 \\ 2176 6639 7654 3549 6989 \\ 4465 6492284}   &  \tabincell{l}{5082 7312 1397 7132 0481 \\ 2175 3530 1491 9036 5403 \\ 9631 9478635}  \\ \hline
    ${\bf A}_8(18,4,6)$ &\tabincell{l}{1532 9290 7353 3720 1326 \\ 6141 0767 7564 4870 0639 \\ 9664 3788 9277952}   &     \tabincell{l}{1532 9290 7353 3720 1326 \\ 6141 0675 9748 0056 0957 \\ 6260 2317 8928128}  \\ \hline
    ${\bf A}_9(18,4,6)$ &\tabincell{l}{1797 3218 5298 8389 5319 \\ 2080 5014 2916 7855 8802 \\ 7766 6966 7098 001534}   &   \tabincell{l}{1797 3218 5298 8389 5319 \\ 2080 5010 3864 5852 2272 \\ 4684 1493 0785 473613}   \\ \hline

    ${\bf A}_2(18,4,9)$ & 5351 9959 2108 4465 545216 & 5350 7797 0493 6727 838720 \\ \hline
    ${\bf A}_3(18,4,9)$ &\tabincell{l}{2297 3939 4211 0816 4158 \\ 5018 7866 3981024}  &\tabincell{l}{2297 3916 8156 5204 1702 \\ 9355 6329 2591869}   \\ \hline
    ${\bf A}_4(18,4,9)$ &\tabincell{l}{2242 2188 4380 9999 7816 \\ 5886 2932 7201 4787 6544 5120}  &  \tabincell{l}{2242 2187 9749 4966 5202 \\ 5921 0903 2573 1520 7125 4016}  \\ \hline
    ${\bf A}_5(18,4,9)$ &\tabincell{l}{2121 9513 6490 9121 3408 \\ 8277 0463 2480 8519 3634 \\ 0332 0312500}   & \tabincell{l}{2121 9513 6267 2348 6969 \\ 6158 8067 3111 6395 4734 \\ 8022 4609375}   \\ \hline
    ${\bf A}_7(18,4,9)$ &\tabincell{l}{7035 1527 6241 0737 9007 \\ 1189 2951 3768 2813 6289 \\ 7842 4922 8045 5523 95936}   &  \tabincell{l}{7035 1527 6232 4593 1444 \\ 0919 3941 7657 2694 5128 \\ 9708 2189 8343 1173 45945}  \\ \hline
    ${\bf A}_8(18,4,9)$ &\tabincell{l}{1053 4207 6388 2296 1833 \\ 5816 7635 2515 7261 2602 \\ 9564 3151 9825 8525 2103 061504}   &     \tabincell{l}{1053 4207 6388 0077 0579 \\ 2562 6795 4570 4178 5576 \\ 6160 5228 1449 8418 7966 849024}  \\ \hline
    ${\bf A}_9(18,4,9)$ &\tabincell{l}{5076 1676 4229 5498 6818 \\ 7238 4546 8292 9652 0878 \\ 2773 0494 1709 3868 1188 \\ 7146 30340}   &   \tabincell{l}{5076 1676 4229 3227 8160 \\ 2958 9753 8502 2100 0195 \\ 2897 8613 0451 1752 6000 \\ 3851 92731}   \\ \hline

    ${\bf A}_2(19,4,6)$ & 4224 2601 1733 4203 3088 & 4224 2601 1357 7889 5040 \\ \hline
    ${\bf A}_3(19,4,6)$ &\tabincell{l}{1050 7720 6899 4195 0274 \\ 6588 7055 2417}  &\tabincell{l}{1050 7720 6899 4192 3508 \\ 9616 2997 7290}   \\ \hline
    ${\bf A}_4(19,4,6)$ &\tabincell{l}{1368 5732 6052 4735 9007 \\ 2223 7673 8963 6291 7888 }  &  \tabincell{l}{1368 5732 6052 4735 9003 \\ 1064 4794 6082 5341 9520}  \\ \hline
    ${\bf A}_5(19,4,6)$ &\tabincell{l}{2716 1071 6124 6620 9890 \\ 3608 6695 3468 8391 7089 515625}   & \tabincell{l}{2716 1071 6124 6620 9890 \\ 3533 6049 3670 9799 7900 062500}   \\ \hline
    ${\bf A}_7(19,4,6)$ &\tabincell{l}{8542 5463 4632 2734 0859 \\ 1485 1388 3155 4865 9869 \\ 1920 1321 9357464}   &  \tabincell{l}{8542 5463 4632 2734 0859 \\ 1485 1324 0784 2886 3871 \\ 7088 4593 7182221}  \\ \hline
    ${\bf A}_8(19,4,6)$ &\tabincell{l}{5023 1019 8748 9776 9685 \\ 9264 4861 8533 3932 5495 \\ 1098 1840 4566 6500608}   &     \tabincell{l}{5023 1019 8748 9776 9685 \\ 9264 4861 2659 3677 6922 \\ 5651 2318 8118 2314496}  \\ \hline
    ${\bf A}_9(19,4,6)$ &\tabincell{l}{1061 3005 8093 3177 5107 \\ 6437 5387 4203 6058 4002 \\ 6449 1270 7651 1789 9237775}   &   \tabincell{l}{1061 3005 8093 3177 5107 \\ 6437 5387 4171 9735 5772 \\ 2490 6664 8070 1103 2667942}   \\ \hline

\end{longtable}

\begin{longtable}{|l@{\extracolsep{\fill}}|l|l|l||}
    \caption{Multilevel type inserting construction I for d=6}\\ \hline
    \label{tab:mutilevel type linkage d6}
    ${\bf A}_q(n,d,k)$ & New & Old\\ \hline  \hline \endfirsthead
    \multicolumn{4}{r}{continued table} \\ \hline
    ${\bf A}_q(n,d,k)$ & New & Old\\ \hline \endhead
    ${\bf A}_2(18,6,9)$ & 9271 5451 7959 0910976 & 9271 5451 5658 5415680 \\ \hline
    ${\bf A}_3(18,6,9)$ &\tabincell{l}{1144 6612 8018 8122 7843 \\ 3650 9436778}  &\tabincell{l}{1144 6612 8018 8113 2295 \\ 9613 3396283}   \\ \hline
    ${\bf A}_4(18,6,9)$ &\tabincell{l}{8507 1058 1461 8280 3382 \\ 5386 3601 4314 848256 }  &  \tabincell{l}{8507 1058 1461 8280 3276 \\ 5044 7701 9755 511808}  \\ \hline
    ${\bf A}_5(18,6,9)$ &\tabincell{l}{1084 2028 9965 7109 7790 \\ 6843 3370 7850 6016 9531 25000}   & \tabincell{l}{1084 2028 9965 7109 7790 \\ 6690 8453 1512 0244 2480 46875}   \\ \hline
    ${\bf A}_7(18,6,9)$ &\tabincell{l}{1742 5150 3388 9755 5131 \\ 8884 9318 2913 8314 8363 \\ 1082 9271 617326}   &  \tabincell{l}{1742 5150 3388 9755 5131 \\ 8884 9225 9937 0807 8414 \\ 3865 7444 402835}  \\ \hline
    ${\bf A}_8(18,6,9)$ &\tabincell{l}{7846 3772 3721 9197 9113 \\ 8381 6353 7233 7401 8480 \\ 6166 8202 2848 43008}   &     \tabincell{l}{7846 3772 3721 9197 9113 \\ 8381 6346 3523 5771 6094 \\ 0063 4183 4301 76768}  \\ \hline
    ${\bf A}_9(18,6,9)$ &\tabincell{l}{1310 0205 1249 3866 3392 \\ 0687 0302 3644 3158 8157 \\ 2950 2470 8175 3276 64556}   &   \tabincell{l}{1310 0205 1249 3866 3392 \\ 0687 0302 3291 8871 4066 \\ 2694 0512 3241 1297 98163}   \\ \hline
\end{longtable}

\newpage
\begin{longtable}{|l@{\extracolsep{\fill}}|l|l|l||}
    \caption{Multilevel type inserting construction II for d=4} \\ \hline
    \label{tab:mutilevel type linkage case2 d4}
    ${\bf A}_q(n,d,k)$ & New & Old\\ \hline  \hline \endfirsthead
    ${\bf A}_q(n,d,k)$ & New & Old\\ \hline \endhead
    ${\bf A}_2(10,4,5)$ & 1178 828 & 1178 824\\ \hline
    ${\bf A}_3(10,4,5)$ &\tabincell{l}{3554 738334  }  &\tabincell{l}{3554 738325 \enspace \enspace}   \\ \hline
    ${\bf A}_4(10,4,5)$ &\tabincell{l}{1105 4718 72592  }  &  \tabincell{l}{1105 4718 72576 \enspace \enspace}  \\ \hline
    ${\bf A}_5(10,4,5)$ &\tabincell{l}{9556 3831 276400  }   & \tabincell{l}{9556 3831 276375 \enspace \enspace}   \\ \hline
    ${\bf A}_7(10,4,5)$ &\tabincell{l}{7983 1695 1903 51258  }   &  \tabincell{l}{7983 1695 1903 51209 \enspace \enspace}  \\ \hline
    ${\bf A}_8(10,4,5)$ &\tabincell{l}{1153 2474 8896 7549 504  }   &     \tabincell{l}{1153 2474 8896 7549 440 \enspace \enspace}  \\ \hline
    ${\bf A}_9(10,4,5)$ &\tabincell{l}{1215 9772 5913 5850 8732  \enspace \enspace  }   &   \tabincell{l}{1215 9772 5913 5850 8651 \enspace \enspace}   \\ \hline

    ${\bf A}_2(16,4,4)$ &\tabincell{l}{8059 6325666}   &   \tabincell{l}{8059 6320222}   \\ \hline
\end{longtable}

\begin{longtable}{|l@{\extracolsep{\fill}}|l|l|l||}
    \caption{Multilevel type inserting construction II for d=6} \\ \hline
    \label{tab:mutilevel type linkage case2 d6}
    ${\bf A}_q(n,d,k)$ & New & Old\\ \hline  \hline \endfirsthead
    \multicolumn{4}{r}{continued table} \\ \hline
    ${\bf A}_q(n,d,k)$ & New & Old\\ \hline \endhead
    ${\bf A}_2(14,6,7)$ & 3453 2242 136 & 3453 2242 120\\ \hline
    ${\bf A}_3(14,6,7)$ &\tabincell{l}{5003 5894 1069 18724}  &\tabincell{l}{5003 5894 1069 18643}   \\ \hline
    ${\bf A}_4(14,6,7)$ &\tabincell{l}{1180 5980 8585 2258 285376}  &  \tabincell{l}{1180 5980 8585 2258 285120}  \\ \hline
    ${\bf A}_5(14,6,7)$ &\tabincell{l}{2910 3849 9692 0980 8789 \\ 39500}   & \tabincell{l}{2910 3849 9692 0980 8789 \\ 38875}   \\ \hline
    ${\bf A}_7(14,6,7)$ &\tabincell{l}{3788 1870 3472 3755 6473 \\ 1907 006636}   &  \tabincell{l}{3788 1870 3472 3755 6473 \\ 1907 004235}  \\ \hline
    ${\bf A}_8(14,6,7)$ &\tabincell{l}{4056 4819 5587 6990 8757 \\ 7560 1388 3904}   &     \tabincell{l}{4056 4819 5587 6990 8757 \\ 7560 1387 9808}  \\ \hline
    ${\bf A}_9(14,6,7)$ &\tabincell{l}{2503 1555 1236 1524 8786 \\ 0765 8765 797556}   &   \tabincell{l}{2503 1555 1236 1524 8786 \\ 0765 8765 790995}   \\ \hline

    ${\bf A}_2(18,6,6)$ & 2829 5832 3494038 & 2829 5832 3460750\\ \hline
    ${\bf A}_3(18,6,6)$ &\tabincell{l}{7977 3414 6743 2777 8633486}  &\tabincell{l}{7977 3414 6743 2776 4264869}   \\ \hline
    ${\bf A}_4(18,6,6)$ &\tabincell{l}{7922 8596 9086 1399 5335 \\ 4258 67868}  &  \tabincell{l}{7922 8596 9086 1399 5334 \\ 3518 63836}  \\ \hline
    ${\bf A}_5(18,6,6)$ &\tabincell{l}{3552 7160 6160 5390 6089 \\ 3919 2138 066570}   & \tabincell{l}{3552 7160 6160 5390 6089 \\ 3916 1618 535195}   \\ \hline
    ${\bf A}_7(18,6,6)$ &\tabincell{l}{3670 3369 3031 7493 1772 \\ 0054 8142 8975 2554 56638}   &  \tabincell{l}{3670 3369 3031 7493 1772 \\ 0054 8142 4227 6535 92745}  \\ \hline
    ${\bf A}_8(18,6,6)$ &\tabincell{l}{2230 0745 3917 5803 6632 \\ 4706 9642 8537 0817 8825 \\ 5096}   &  \tabincell{l}{2230 0745 3917 5803 6632 \\ 4706 9642 8501 8972 8194 \\ 8024}  \\ \hline
    ${\bf A}_9(18,6,6)$ &\tabincell{l}{6362 6854 5986 5481 8930 \\ 7460 2002 4319 4539 0489 \\ 669738}   &  \tabincell{l}{6362 6854 5986 5481 8930 \\ 7460 2002 4317 3949 8970 \\ 153871}  \\ \hline

    ${\bf A}_2(19,6,6)$ & 4527 3330 8765 3534 & 4527 3330 8758 6958\\ \hline
    ${\bf A}_3(19,6,6)$ &\tabincell{l}{6461 6465 8861 9087 6128 \\ 77754}  &\tabincell{l}{6461 6465 8861 9087 5697 \\ 71903}   \\ \hline
    ${\bf A}_4(19,6,6)$ &\tabincell{l}{2028 2520 8086 0518 1180 \\ 3609 1661 5452}  &  \tabincell{l}{2028 2520 8086 0518 1180 \\ 3566 2059 9324}  \\ \hline
    ${\bf A}_5(19,6,6)$ &\tabincell{l}{2220 4475 3850 3369 1301 \\ 9528 8801 5208 79570}   & \tabincell{l}{2220 4475 3850 3369 1301 \\ 9528 8648 9232 22695}   \\ \hline
    ${\bf A}_7(19,6,6)$ &\tabincell{l}{8812 4789 6969 2301 1184 \\ 5835 5714 5620 5845 7352 \\ 1554}   &  \tabincell{l}{8812 4789 6969 2301 1184 \\ 5835 5714 5587 3513 6047 \\ 4303}  \\ \hline
    ${\bf A}_8(19,6,6)$ &\tabincell{l}{9134 3853 1246 4091 8046 \\ 5999 2858 2456 0473 4014 \\ 9168504}   &  \tabincell{l}{9134 3853 1246 4091 8046 \\ 5999 2858 2455 7658 6409 \\ 8711928}  \\ \hline
    ${\bf A}_9(19,6,6)$ &\tabincell{l}{4174 5579 3021 7742 6700 \\ 4624 6291 8490 3567 4142 \\ 3743 702202}   &  \tabincell{l}{4174 5579 3021 7742 6700 \\ 4624 6291 8490 3548 8840 \\ 0068 059399}  \\ \hline
\end{longtable}

\newpage

\begin{longtable}{|l@{\extracolsep{\fill}}|l|l|l||}
    \caption{Multilevel type inserting construction II for d=8} \\ \hline
    \label{tab:mutilevel type linkage case2 d8}
    ${\bf A}_q(n,d,k)$ & New & Old\\ \hline  \hline \endfirsthead
    \multicolumn{4}{r}{continued table} \\ \hline
    ${\bf A}_q(n,d,k)$ & New & Old\\ \hline \endhead
    ${\bf A}_2(18,8,9)$ & 1801 5215 3991 16904 & 1801 5215 3991 16872\\ \hline
    ${\bf A}_3(18,8,9)$ &\tabincell{l}{5814 9739 3804 1767 0685  308590}  &\tabincell{l}{5814 9739 3804 1767 0685  308347}   \\ \hline
    ${\bf A}_4(18,8,9)$ &\tabincell{l}{3245 1855 3767 8429 8642 \\ 4312 3978 79872}  &  \tabincell{l}{3245 1855 3767 8429 8642 \\ 4312 3978 78848}  \\ \hline
    ${\bf A}_5(18,8,9)$ &\tabincell{l}{5551 1151 2317 3587 8357 \\ 9602 1219 6595 692000}   & \tabincell{l}{5551 1151 2317 3587 8357 \\ 9602 1219 6595 688875}   \\ \hline
    ${\bf A}_7(18,8,9)$ &\tabincell{l}{4318 1145 6739 6591 8176 \\ 2301 6095 3650 9153 0833 \\ 954554}   &  \tabincell{l}{4318 1145 6739 6591 8176 \\ 2301 6095 3650 9153 0833 \\ 937747}  \\ \hline
    ${\bf A}_8(18,8,9)$ &\tabincell{l}{5846 0065 4932 3635 8379 \\ 3403 4302 9250 8651 1686 \\ 9853 92640}   &     \tabincell{l}{5846 0065 4932 3635 8379 \\ 3403 4302 9250 8651 1686 \\ 9853 59872}  \\ \hline
    ${\bf A}_9(18,8,9)$ &\tabincell{l}{3381 3919 1352 2728 4246 \\ 2028 0247 0185 2687 1078 \\ 7157 1285 4492}   &   \tabincell{l}{3381 3919 1352 2728 4246 \\ 2028 0247 0185 2687 1078 \\ 7157 1279 5443}   \\ \hline
\end{longtable}

\end{document}